\documentclass[11pt, a4paper]{article}

\usepackage{amsmath, amsthm, amssymb, amsfonts}
\usepackage{lmodern} 
\usepackage[utf8]{inputenc}
\usepackage{mathtools}
\usepackage{hyperref}
\usepackage{graphicx}
\usepackage{dsfont}
\usepackage{fullpage}
\usepackage{color}
\usepackage{soul} 
\usepackage[title]{appendix}
\usepackage{tikz}
\usetikzlibrary{patterns}
\usetikzlibrary{shapes.multipart}
\usetikzlibrary{arrows}

\definecolor{labelkey}{cmyk}{.4,.2,0,0}

\newcommand{\be}{\begin{equation}}
\newcommand{\ee}{\end{equation}}
\newcommand{\bea}{\begin{eqnarray}}
\newcommand{\eea}{\end{eqnarray}}

\newcommand{\R}{\ensuremath{\mathbb{R}}}

\newcommand{\Z}{\ensuremath{\mathbb{Z}}}

\renewcommand{\rho}{\varrho}

\newcommand{\eps}{\varepsilon}
\renewcommand{\leq}{\leqslant}
\renewcommand{\geq}{\geqslant}
\renewcommand{\le}{\leqslant}
\renewcommand{\ge}{\geqslant}

\newcommand{\Gammainv}{\mathrm{Gamma}^{-1}}
\newcommand{\xx}{\mathbf x}
\newcommand{\yy}{\mathbf y}
\newcommand{\ZZ}{\mathbf{Z}}
\newcommand{\Zcal}{\mathcal{Z}}

\usepackage{titlesec}
\titleformat{\section}{\large\bf}{\thesection}{1em}{}
\titleformat{\subsection}[runin]{\bf}{\thesubsection}{1em}{}[.]
\titleformat{\subsubsection}[runin]{\it}{\thesubsubsection}{1em}{}[.]

\newtheorem{theorem}{Theorem}[section]

\newtheorem{lemma}[theorem]{Lemma}
\newtheorem{proposition}[theorem]{Proposition}

\theoremstyle{definition}

\theoremstyle{definition}
\newtheorem{definition}[theorem]{Definition}

\theoremstyle{definition}

\theoremstyle{definition}
\newtheorem{remark}{Remark}[section]

\usepackage{authblk}
\author[1]{Guillaume Barraquand and Pierre Le Doussal}
\affil[1]{\normalsize Laboratoire de Physique de l'\'Ecole Normale Sup\'erieure, ENS, Universit\'e PSL, CNRS, Sorbonne Universit\'e, Universit\'e Paris-Cité, 75005 Paris, France}

\title{\bf \large A stationary model of non-intersecting directed polymers}
\date{}

\begin{document}

\maketitle

\begin{abstract} 
We consider the partition function $Z_{\ell}(\vec x,0\vert \vec y,t)$ of $\ell$ non-intersecting continuous directed polymers of length $t$ in dimension $1+1$, in a white noise environment, starting from positions $\vec x$ and terminating at positions $\vec y$. When $\ell=1$, it is well known that for fixed $x$, the field $\log Z_1(x,0\vert y,t)$ solves the Kardar-Parisi-Zhang equation and admits the Brownian motion as a stationary measure. In particular,  as $t$ goes to infinity, $Z_1(x,0\vert y,t)/Z_1(x,0\vert 0,t) $ converges to the exponential of a Brownian motion $B(y)$. In this article, we show an analogue of this result for any $\ell$. We show that $Z_{\ell}(\vec x,0\vert \vec y,t)/Z_{\ell}(\vec x,0\vert \vec 0,t) $ converges as $t$ goes to infinity to an explicit functional $Z_{\ell}^{\rm stat}(\vec y)$ of $\ell$ independent Brownian motions. This functional $Z_{\ell}^{\rm stat}(\vec y)$ admits a simple description as  the partition sum for $\ell$ non-intersecting semi-discrete polymers on $\ell$ lines. We discuss applications to the endpoints and midpoints distribution for long non-crossing polymers and derive explicit formula in the case of two polymers. To obtain these results, we show that the stationary measure of the O'Connell-Warren multilayer stochastic heat equation is given by a collection of independent Brownian motions.  This in turn is  shown via analogous results in a discrete setup for the so-called log-gamma polymer and exploit the connection between non-intersecting log-gamma polymers and the geometric RSK correspondence  found in \cite{corwin2014tropical}.

\end{abstract}

\setcounter{tocdepth}{1}
\tableofcontents

\subsection*{Motivation} 
Interacting directed elastic lines in presence of quenched point disorder appear in many physical systems in dimension $d=2+1$ and $d=1+1$. Examples are
pinned vortex phases in superconductors \cite{blatter1994vortices}, such as the Bragg glass and the vortex glass
\cite{doussal2011novel}, arrays of dislocation lines in solids \cite{moretti2004depinning}, surfaces of crystals with quenched disorder \cite{toner1990super}. In dimension $d=1+1$ the statistical mechanics is more tractable,
and experiments are possible \cite{bolle1999observation,einstein2014dynamical,einstein2007using}.
One can define a phase (or counting) field and, for generic interactions between the lines, the system can be described by
a Sine-Gordon model with a quenched random phase \cite{nattermann1991flux}. Being known as the Cardy-Ostlund  model \cite{cardy1982random},
it describes the classical 2D XY model with quenched random fields (and vortices excluded), and is
also related to random bond 2D dimer models \cite{zeng1996ground,zeng1999thermodynamics,bogner2004test}. 
Renormalization group and numerical studies of the Cardy-Ostlund model \cite{cardy1982random,toner1990super,hwa1994vortex,zeng1996ground,le2007disordered,perret2012super,ristivojevic2012super} 
showed that there are two phases, a high temperature rough phase and a low temperature super-rough glass phase. 
A particular case arises when the repulsive line interactions are replaced by non-crossing constraints. The 
pure system can then be modeled by free fermions \cite{gennes1968soluble,pokrovsky1979ground}.  
In presence of space-time disorder, non-interacting fermions are also believed to belong to the universality class of the Cardy-Ostlund model \cite{guruswamy2000gl,le2007disordered}. 

\medskip 
In another direction, the system of non-crossing lines in a random potential was studied more directly using the replica method and the Bethe ansatz.
Using a mapping of the problem onto $U(n)$ fermions with attractive interactions \cite{kardar1987replica}, 
the quenched-averaged free energy was calculated in \cite{kardar1987replica} from the continuation to $n=0$.
In \cite{emig2001probability,emig2000thermodynamic} the mesoscopic fluctuations of the free energy were 
predicted using an additional effective mapping to a weakly repulsive Bose gas.
Attempts to calculate the density correlations of the line lattice 
from the excitation spectrum of the $U(n)$ fermions
led to conflicting results \cite{balents1993system,tsvelik1992influence}. 
More recently, the nested Bethe ansatz was used to obtain non-crossing probabilities 
\cite{de2015crossing,de2016crossing} as well as the tail of the probability density function (PDF) of the free energy of
$N$ non-crossing lines \cite{de2017mutually}. These later works are based in a large part on 
recent progress in the mathematical literature to describe non-crossing continuous directed
polymers \cite{borodin2014macdonald} (see also \cite{o2016multi, johnston2020scaling}).

\section{Model and main results}
\label{sec:intro}

\subsection{The directed polymer model} Before describing the aim and main results of the present paper, let us recall the model of a single directed elastic line,
also called a directed polymer, in a random potential  in dimension $d=1+1$. For a polymer of length $t$, with endpoints in $x,y$, it is
defined by the partition sum 
\begin{equation}
\label{eq:defpointopointcontinuousDP}
Z(x,0\vert y,t) = \int_{x(0)=x}^{x(t)=y}  Dx(\tau) e^{-  \int_0^t dt \left( \frac{1}{2} (\frac{dx(\tau)}{d\tau})^2 - \eta(x(\tau),\tau)\right) }
= \mathrm{E} \left[ e^{ \int_0^t d\tau  \eta(W(\tau),\tau)} \right] p_t(x,y)
\end{equation}
where  in the last (equivalent) expression, $p_t(x,y) =\frac{e^{- \frac{(x-y)^2}{2 t} }}{\sqrt{2 \pi t}} $ is the heat kernel, and the expectation $\mathrm E$ is over standard Brownian bridges $W$, with endpoints 
$W(0)=x$ and $W(t)=y$. Here $\eta(x,\tau)$ is a Gaussian noise, that we assume to be white in  the (fictitious) time variable $\tau$.  So far, we do not make assumptions on its correlations in space, but we focus below on the case of delta correlations
in space, that is $\mathbb{E}[\eta(x,t),\eta(x',t')]=\delta(x-x')\delta(t-t')$, which
requires extra mathematical care (see Section \ref{sec:KPZlimit}). From the Feynman-Kac formula, the directed polymer partition sum $Z=Z(x,0\vert y,t)$ 
is the solution of the 
stochastic heat equation (SHE) in the It\^o sense
\begin{equation}
\partial_t Z=\frac 1 2 \partial_{yy} Z+ \eta Z 
\label{eq:SHE}
\end{equation}
with initial condition $Z(x,0\vert y,t=0)=\delta(y-x)$. 
The height field defined by (minus) the free energy of the polymer as a function of the endpoint position, that is $h(y,t) = \log Z(y,t)$, solves \cite{KPZ} the one-dimensional  Kardar-Parisi-Zhang (KPZ) equation 
(in that case with the so-called droplet initial condition). The KPZ equation describes the stochastic growth of an interface of
height $h(y,t)$ in presence of a noise $\eta$.

\medskip 
A well known property of the KPZ equation on the full line with space-time white noise $\eta$ is 
that it is associated to a Brownian stationary measure. More precisely, 
while the global height grows linearly with time with non trivial $t^{1/3}$ fluctuations, 
the height differences between any two points, will reach at large time a stationary distribution given by a Brownian motion.
This was pointed out long ago \cite{forster1977large,huse1985huse, parisi1990replica} in the physics literature. The fact that the Brownian measure is stationary is proved rigorously 
in \cite{bertini1995stochastic} (see also \cite{funaki2015kpz}). In terms of directed polymers, ratios of partition sums with different endpoints 
 converge in the limit of very long polymers. They have the same distribution as  ratios of the exponential of a Brownian motion, which we denote by  $Z_1^{\rm stat}(y)=e^{B(y)}$. More precisely, 
\begin{equation}\label{singleDP}
\lim_{t \to +\infty} \frac{Z(x,-t \vert y,0)}{Z(x,- t \vert z,0)} \overset{(d)}{=} \frac{Z_1^{\rm stat}(y)}{Z_1^{\rm stat}(z)} = e^{B(y)-B(z)}
\end{equation}
for any fixed $x$, where $B(y)$ is a standard Brownian motion (i.e. with $B(0)=0$ and $dB(y)^2=dy$). In \eqref{singleDP}, the limit on the left-hand-side should exist for almost every realization of the environment $\eta$, and the equality in distribution  holds as a process, that is jointly in $y$ and $z$ (see \cite[Theorem 1.11]{das2022localization} for a rigorous proof of the convergence in distribution). The particular realization of $B$ depends of course in a very non-trivial way on the random potential $\eta$. 

\medskip 
The process $Z_1^{\rm stat}(y)=e^{B(y)}$ describes the (unnormalized) endpoint measure for very long polymers.
Although it is not normalizable in the full space, we may apply an external force on the directed polymer endpoint, as can be done in experiments. Assuming the force derives from a sufficiently confining potential $U(y)$, the density of the polymer endpoint becomes proportional to $Z_1^{\rm stat}(y) e^{- \beta U(y)}$, which is now normalizable. Similarly, the probability measure of the {\it midpoint} of a very long polymer can be read from \eqref{singleDP} by considering 
$Z(x,-t \vert y,0) Z(y,0\vert x,t)$ for large $t$ and it is distributed as $e^{B(y)+ \tilde B(y)}$, where $B$ and $\tilde B$ are two independent standard Brownian motions.

\medskip 

\subsection{Non-intersecting directed polymers} A very natural question is to consider now $\ell$ continuum directed polymers constrained not to cross, 
in the same random
potential $\eta$, and ask what is the stationary measure. In this paper we obtain this stationary measure for the case of space-time white noise, in a form
which generalizes the result for a single polymer $\ell=1$, depending now on $\ell$ Brownian motions. This provides information on the behavior of the endpoint probabilities of very long non-crossing polymers, as well as their midpoint probabilities, as we discuss below. 
The result is obtained by considering a discrete version, the log gamma polymer model, recalled below, 
for which we also obtain novel results. 
We first summarize our main results in the continuum 
and explain later the general idea of the method, as well as the new results in the discrete setting.

\medskip 
To express our result in the next subsection, leaving aside mathematical subtleties for now, let us define $Z_{\ell}(\vec x,0 \vert \vec y ,t)$,
the partition function of $\ell$ non intersecting continuous polymers, with starting points 
at $\vec x=(x_1 < x_2<\dots <x_\ell)$ and ending at $\vec y=(y_1<y_2<\dots < y_\ell)$. It can be 
defined as the expectation over $\ell$ non-intersecting Brownian paths $W_j(\tau)$, $j=1,\dots,\ell$
starting from $W_j(0)= x_j$ and ending at $W_j(t)=y_j$,  
\begin{equation}
Z_{\ell}(\vec x,0 \vert \vec y ,t) = \det\left(  p_t(x_i,y_j) \right)_{i,j=1}^{\ell}\mathrm{E} \left[ e^{-  \int_0^t d\tau \sum_j \eta(W_j(\tau),\tau)} \right], 
\label{eq:defZl}
\end{equation} 
where the expectation $\mathrm E$ is taken over $W_1,\dots W_\ell$, $\ell$ non-intersecting Brownian bridges. The first factor corresponds to the standard formula for the propagator of $\ell$ non-intersecting Brownian motions. The partition function $Z_{\ell}(\vec x,0 \vert \vec y ,t)$  was considered in physics \cite{emig2001probability,de2015crossing,de2017mutually}
and is given by the Karlin Mc Gregor formula \cite{karlin1959coincidence}
\begin{equation}
Z_{\ell}(x_1, \dots,x_{\ell};0\vert y_1, \dots, y_{\ell};t) = \det\left(Z_1(x_i,0 \vert y_j,t)\right)_{i,j=1}^{\ell}.
\label{eq:defZKarlinMcGregor}
\end{equation}
From a mathematical point of view, \eqref{eq:defZKarlinMcGregor} can be taken as a definition of $Z_{\ell}$, which poses no issue even in the case of a white noise in space. 
The partition function $Z_{\ell}$ satisfies the stochastic partial differential equation
\begin{equation}
\partial_t Z_\ell= \sum_i \frac 1 2 \partial_{y_iy_i} Z_\ell +  Z_\ell  \sum_i \eta(y_i,t) 
\label{eq:SPDE}
\end{equation} 
on the Weyl chamber 
\begin{equation} 
\mathbb W_{\ell} = \left\lbrace \vec y\in \mathbb R^{\ell};  y_1<y_2<\dots < y_\ell\right\rbrace,
\label{eq:defWeyl}
\end{equation} 
with the boundary condition that $Z_\ell=0$ whenever any $y_j = y_{j+1}$. 
The formula \eqref{eq:defZl} is well defined for a smooth spatial noise correlator,  but there are some mathematical issues in defining the white noise
limit. We refer to \cite{o2016multi, corwin2017intermediate, chandra2019local} for a more precise  mathematical discussion on the well-posedness of such stochastic PDEs.

\subsection{Main results}
\label{sec:mainresults}
 We can now state our main result for the continuum model described in the previous section. We find that the generalization of \eqref{singleDP} to $\ell$ non-crossing directed polymers is the following. For $\vec x, \vec y, \vec z \in \mathbb W_{\ell}$, 
\begin{equation}
\lim_{t \to +\infty} 
 \frac{ Z_{\ell}(\vec x;-t\vert \vec y ;0) }{Z_{\ell}(\vec x;-t\vert \vec z ;0) } \overset{(d)}{=}
\frac{Z^{\rm stat}_\ell(\vec y)}{Z^{\rm stat}_\ell(\vec z)} 
\label{eq:multipolymers}
\end{equation}
where the partition function $Z^{\rm stat}_\ell$ is now defined by 
\begin{equation}
Z^{\rm stat}_\ell(\vec y)= \int_{GT(\vec y)} \prod_{k=1}^{\ell} \prod_{i=1}^{k} e^{
	 B_{\ell-k+1}(z_i^{k})- B_{\ell-k+1}(z_{i-1}^{k-1})}\prod_{k=1}^{\ell-1} \prod_{i=1}^{k} \mathrm d z_i^{k},
	 \label{eq:defZstat}
\end{equation}
where the $B_k(z)$, $k=1,\dots,\ell$ are $\ell$ independent standard Brownian motions (we will assume for convenience that $B_k(0)=0$, although it does not matter since   \eqref{eq:defZstat} only involves increments of the Brownian motions). We use the
convention that $z_0^k=0$ for $0\le k\le \ell$, and the integration is performed on the interlaced set of $\ell(\ell-1)/2$ independent auxiliary variables (which form a Gelfand-Tsetlin pattern) 
\begin{equation} GT(\vec y)  = \lbrace (z_i^k)_{1\le i\le k\le \ell }: z_{i}^{k+1}\le z_i^k \le z_{i+1}^{k+1}\text{ for }1\le i\le k\le \ell-1, \text{ and } z_i^{\ell}=y_i\text{ for }1\le i\le \ell  \rbrace.
\label{eq:defGT}
\end{equation} 
Note that for any $x\leq y_1$, 
\begin{equation} 
Z_{\ell}^{\rm stat}(\vec y) = e^{\sum_{i=1}^{\ell} B_i(x)}Z_{\ell}^{\rm stat}(\vec y-x\vec 1).
\label{eq:translation}
\end{equation}

\begin{remark} 
\label{rem:t2/3scale}
We expect that as $t$ goes to infinity, 
\begin{equation}
 \frac{ Z_{\ell}(\vec x;-t\vert \vec y ;0) }{Z_{\ell}(\vec x;-t\vert \vec z ;0) } \overset{(d)}{\simeq}
\frac{Z^{\rm stat}_\ell(\vec y)}{Z^{\rm stat}_\ell(\vec z)} 
\end{equation}
  holds not only for fixed $\vec y$ and $\vec z$, but as long as $ \vert y_i -z_j \vert \ll t^{2/3}$ for all $1\leq i,j\leq \ell$. 
\end{remark}

\subsection{Outline}
 In the remainder of this Section, we discuss a  graphical interpretation of the quantity $Z_{\ell}^{\rm stat}(\vec y)$ that we can view as the partition function of en ensemble of semi-discrete directed polymers. We also discuss generalizations of our main result \eqref{eq:multipolymers} and applications to the midpoint and the  endpoint distribution of $n$ non-intersecting polymers.  

In Section \ref{sec:discretecase}, we consider a special exactly solvable model of discrete directed polymers on the $\mathbb Z^2$ lattice (originally introduced in \cite{seppalainen2012scaling}) with inverse gamma distributed weights. We state and prove a discrete analogue of our main result using a connection, developed in \cite{corwin2014tropical},  between polymer partition functions and the geometric RSK correspondence. Then, in Section \ref{sec:KPZlimit}, we take the continuous limit of the results from Section \ref{sec:discretecase} to derive our main result \eqref{eq:multipolymers} and the other results given below in Section \ref{sec:intro}. We provide actually two routes leading to the same result: in the first derivation, we start finding the (simpler) stationary measure of an object called the multilayer stochastic heat equation introduced in \cite{o2016multi} (Section \ref{sec:invariantOConnellWarren}). We then deduce the stationary measure of the stochastic PDE \eqref{eq:SPDE} using a remarkable relation from \cite{o2016multi} between the partition functions of non-intersecting polymers with distinct endpoints and the multilayer stochastic heat equation. In the second derivation, we do not appeal to \cite{o2016multi} but use a discrete analogue of \eqref{eq:multipolymers}, and we take a scaling limit. 

\subsection{Graphical interpretation, generalizations and applications} 
\subsubsection{Graphical interpretation} For $\ell=2$ one has, with $z=z_1^1$ and $y_1<y_2$
\begin{equation}\label{Z2n}
Z^{\rm stat}_2(y_1,y_2)= e^{B_1(y_1)+B_1(y_2)} \int_{y_1}^{y_2} dz e^{-B_1(z) + B_2(z)},
\end{equation}
which can be seen as the partition function of two non-intersecting semi-discrete polymers as shown in Fig. \ref{fig:OY}. The polymer paths live on two horizontal lines. The two polymers start at the horizontal coordinate $0$, and end on the second line at horizontal coordinates $y_1$ and $y_2$ respectively. The energy collected by a given polymer path is the sum of increments of two independent standard Brownian motions $B_1$ and $B_2$ along each horizontal line.  

\medskip 
For $\ell=3$ and for $y_1<y_2<y_3$ one has
\begin{multline} \label{Z3n}
 Z^{\rm stat}_3(y_1,y_2,y_3) = e^{B_1(y_1)+B_1(y_2)+B_1(y_3)} \\
\times \int_{y_1<z^2_1<y_2<z^2_2<y_3} dz^2_1 dz^2_2 
e^{B_2(z_1^2)-B_1(z_1^2) + B_2(z_2^2)-B_1(z_2^2)} 
\int_{z^2_1 < z^1_1 < z^2_2} dz^1_1 e^{B_3(z_1^1)- B_2(z_1^1) }.
\end{multline}
This can also be interpreted as the partition function of three non-intersecting semi-discrete polymers as shown in Fig. \ref{fig:OY} (right).

\medskip 
More generally, for any $\ell$, $Z_\ell^{\rm stat}(\vec y)$ is the partition function of $\ell$ non-intersecting semi-discrete polymers on $\ell$ horizontal lines, indexed from top to bottom,  weighted by independent standard Brownian motions $B_1, \dots, B_{\ell}$. The polymer paths all start from the horizontal coordinate $0$, and end on the first line at locations $y_1, \dots, y_{\ell}$. Such a semi-discrete polymer model has been well-studied, starting from \cite{o2001brownian}, and is generally referred to as the O'Connell-Yor polymer model.  In earlier references, consideration was restricted to the case where semi-discrete non-intersecting polymers arrive all at the same point. It seems remarkable that partition functions of the semi-discrete O'Connell-Yor polymer describe the stationary measure of non-intersecting Brownian polymers.

 Note that we have represented in Fig. \ref{fig:OY} the case where all $y_i>0$ but this is sufficient, since the
origin can be chosen arbitrarily far to the left, as it cancels in the partition function ratio.

\begin{figure}
	\centering
\begin{center}	
	\begin{tikzpicture}
	\begin{scope}
	\draw[gray] (0,0) -- (6,0) node[anchor=west] {$B_2$};
\draw[gray] (0,1) -- (6,1) node[anchor=west] {$B_1$};
\fill (3,0) circle(0.05);
\draw (3,0) node[anchor=north] {$z_{1}^1$};
\draw (2,0.9) -- (2,1.1) node[anchor=south] {$y_1$};
\draw (5,0.9) -- (5,1.1) node[anchor=south] {$y_2$};
\draw[->] (3,0) -- (3,0.5);
\draw (3,0) -- (3,1);
\draw[ultra thick] (0,0) -- (3,0);
\draw[ultra thick] (0,1) -- (2,1);
\draw[ultra thick] (3,1) -- (5,1);
\draw (0,-0.1) node[anchor=north] {$0$} -- (0,0.1);
\draw (0,0.9) node[anchor=north] {$0$} -- (0,1.1);
	\end{scope}
	\begin{scope}[xshift=7.5cm, yshift=-0.5cm]
\draw[gray] (0,0) -- (7,0) node[anchor=west] {$B_3$};
\draw[gray] (0,1) -- (7,1) node[anchor=west] {$B_2$};
\draw[gray] (0,2) -- (7,2) node[anchor=west] {$B_1$};
\draw (3,0) node[anchor=north] {$z_1^1$};
\draw (2,1) node[anchor=north] {$z_1^2$};
\draw (5,1) node[anchor=north] {$z_2^2$};
\draw (4,1.9) -- (4,2.1) node[anchor=south] {$y_2$};
\draw (6,1.9) -- (6,2.1) node[anchor=south] {$y_3$};
\draw (1,1.9) -- (1,2.1) node[anchor=south] {$y_1$};
\draw[->] (3,0) -- (3,0.5);
\draw (3,0) -- (3,1);
\draw[->] (2,1) -- (2,1.5);
\draw (2,1) -- (2,2);
\draw[->] (5,1) -- (5,1.5);
\draw (5,1) -- (5,2);
\draw[ultra thick] (0,0) -- (3,0);
\draw[ultra thick] (0,1) -- (2,1);
\draw[ultra thick] (3,1) -- (5,1);
\draw[ultra thick] (0,2) -- (1,2);
\draw[ultra thick] (2,2) -- (4,2);
\draw[ultra thick] (5,2) -- (6,2);
\draw (0,-0.1) node[anchor=north] {$0$} -- (0,0.1);
\draw (0,0.9) node[anchor=north] {$0$} -- (0,1.1);
\draw (0,1.9) node[anchor=north] {$0$} -- (0,2.1);
	\end{scope}
	\end{tikzpicture}
	\end{center}
	\caption{The partition function $Z_{\ell}^{\rm stat}(\vec y) $ defined in \eqref{eq:defZstat} can be interpreted as the partition sum of non-crossing semi-discrete polymers. We have depicted the case $\ell=2$ (left) and $\ell=3$ (right).}
	\label{fig:OY}
\end{figure}
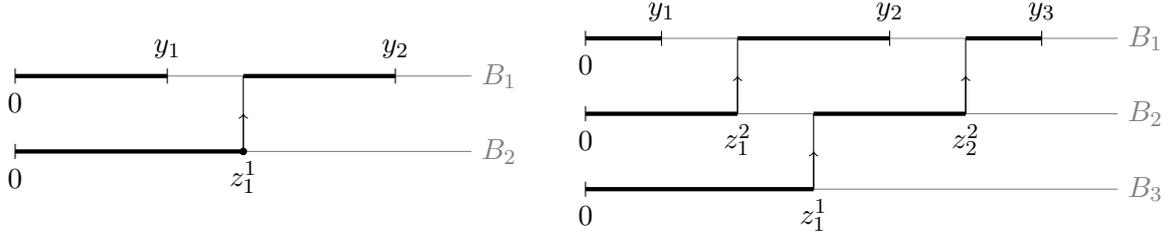

\medskip 
Since, as we just discussed, $Z_{\ell}^{\rm stat}(\vec y)$ can be seen as the partition function of $\ell$ non-intersecting polymers, the Karlin-McGregor formula yields the alternative formula 
\begin{equation} \label{KMGOC} 
Z_{\ell}^{\rm stat}(\vec y) = \det \left( Z^{\rm OY}[(0,i)\vert (y_j, 1)] \right)_{1\leqslant i,j\leqslant \ell},
\end{equation}
where $Z^{\rm OY}[(0,i)\vert (y, 1)]$ is the partition function for a single semi-discrete polymer starting from horizontal coordinate $0$ on the $i$-the line (lines are indexed from top to bottom) and ending on the first line at horizontal coordinate $y$ (see Section \ref{sec:OYdetails} below for more details about the semi-discrete O'Connell-Yor polymer mode, see in particular  \eqref{eq:defZOY} for an explicit expression defining $Z^{\rm OY}[(0,i)\vert (y_j, 1)$).

\begin{remark}
The quantity $Z_{\ell}^{\rm stat}$ defined in \eqref{eq:defZstat}  satisfies a surprising identity. For any realization of the Brownian motions $B_1, \dots, B_{\ell}$,  $Z_{\ell}^{\rm stat}$ remains unchanged if we replace the Brownian motions $B_1, \dots, B_{\ell}$ by a family of processes $\mathcal WB_1, \dots, \mathcal WB_{\ell}$ built from the Brownian motions $B_1, \dots, B_{\ell}$ (this process is called the melon of the Brownian motions $B_1, \dots, B_{\ell}$ in \cite{dauvergne2018directed}, where a similar invariance plays a crucial role in the construction of the Airy sheet). We explain the details of this construction in Section \ref{sec:identity}, based on results from \cite{corwin2020invariance}. Remarkably, the distribution of the process $(\mathcal WB_1, \dots, \mathcal WB_{\ell})$ converges at large scale to non-intersecting Brownian motions, that is to a GUE Dyson Brownian motion $(\Lambda_1(y), \dots, \Lambda_{\ell}(y))$. It remains to be seen whether these observations can help analyzing the distribution of $Z_{\ell}^{\rm stat}$.
\end{remark}

\subsubsection{Generalizations} 
\label{sec:generalizations}
We may consider the case when both endpoints and starting points vary.  In that case, we obtain 
\begin{equation}
\lim_{t \to +\infty} 
 \frac{ Z_{\ell}(\vec x;-t\vert \vec y ;0) }{Z_{\ell}(\vec r;-t\vert \vec z ;0) } \overset{(d)}{=}
\frac{Z^{\rm stat}_\ell(\vec y)}{Z^{\rm stat}_\ell(\vec z)}  \frac{\tilde Z^{\rm stat}_\ell(\vec x)}{\tilde Z^{\rm stat}_\ell(\vec r)}, 
\label{eq:multipolymersbothends}
\end{equation}
where $Z_{\ell}^{\rm stat}$ is defined with respect to Brownian motions $B_1, \dots, B_{\ell}$ and $\tilde Z_{\ell}^{\rm stat}$ is defined with respect to another independent set of  Brownian motions $\tilde B_1, \dots, \tilde B_{\ell}$.

It is also natural to consider now very long non intersecting polymers 
which come from different directions, i.e. with starting points at time $-t$ equal to $x_i=a_i t$, where $a_1\le  \dots \le a_{\ell}$.
In that case we conjecture that \footnote{when all drifts are identical, $a_i=a$, this can be deduced by taking a continuous limit of the results in Section \ref{sec:sendingwithanangle}.}
\begin{equation}
\lim_{t \to +\infty} 
 \frac{ Z_{\ell}(\vec x;-t\vert \vec y ;0) }{Z_{\ell}(\vec x;-t\vert \vec z ;0) } = 
\frac{Z^{\rm stat}_\ell(\vec y; \vec a)}{Z^{\rm stat}_\ell(\vec z; \vec a)} 
\label{eq:mainresultwithdrifts}
\end{equation}
where $Z^{\rm stat}_\ell(\vec y; \vec a)$ is defined as in \eqref{eq:defZstat} except that now the Brownian motion $B_i$ have drifts $a_{i}$. We also conjecture that the processes $Z^{\rm stat}_\ell(\vec y; \vec a)$ constitute the set of all extremal stationary measures of the stochastic PDE \eqref{eq:SPDE}.

\subsubsection{Limits} 

There are two natural limits of the stationary partition function $Z^{\rm stat}_\ell(\vec y)$ to consider, corresponding
respectively to its short distance behavior (dense limit) and its large distance behavior (dilute limit).

Consider the short distance limit where all the $y_j$ are almost equal, $y_j \approx y$. In that case 
one can approximate all the Brownian motions $B_i(z_i^\ell)$  in \eqref{eq:defZstat} by $B_i(y)$, and one gets 
\be 
Z^{\rm stat}_\ell(\vec y) \simeq \prod_{i=1}^{\ell}\frac{e^{ B_i(y)}}{(i-1)!} \Delta(\vec y)
\ee
where $\Delta(\vec y):=\prod_{i<j}( y_j-y_i)$,  as easily obtained from \eqref{KMGOC} using \eqref{eq:translation} and the fact that 
$Z^{\rm OY}[(0,i)\vert (y_j-y_1, 1)]\simeq \frac{(y_j-y_1)^{i-1}}{(i-1)!}$. 
The Vandermonde determinant $\Delta(\vec y)$ is indeed the equilibrium measure for $\ell$ non-crossing Brownian motions (i.e. directed polymers in the absence of
random potential). Equivalently, the Vandermonde is a harmonic function for the Laplacian with vanishing  boundary condition on the border of the Weyl chamber $\mathbb W_{\ell}$.

Consider now the large distance limit, that is the case where all $y_i-y_j \gg 1$. The integrals of exponentials in \eqref{eq:defZstat} are dominated by the maximum over the Gelfand-Tsetlin pattern. More precisely one has
\be 
    \lim_{x\to\infty} \frac{1}{\sqrt{x}} \log Z^{\rm stat}_{\ell}(x\vec y) = \sup_{z \in GT(\vec y)} \left\lbrace \sum_{k=1}^{\ell}\sum_{i=1}^{k} B_{\ell-k+1}(z_i^k)- B_{\ell-k+1}(z_{i-1}^{k-1})  \right\rbrace.
    \label{eq:supremum}
\ee

\subsubsection{The case of two non-intersecting polymers} When $\ell=2$, \eqref{eq:defZstat} can be rewritten as, see \eqref{Z2n}
\be \label{Z2stat}
Z^{\rm stat}_2(y_1,y_2) = e^{B_1(y_1)+B_2(y_1)} Z^{\rm OY}[(y_1,2) \vert (y_2,1)]  
\ee 
where $Z^{\rm OY}((y_1,2) \vert (y_2,1))$ is, again, a semi-discrete polymer partition function (see Fig. \ref{fig:OY} (left)). The random variable $B_1(y_1)+B_2(y_1)$ is independent from  $Z^{\rm OY}[(y_1,2) \vert (y_2,1)]$, whose distribution depends only on the difference $y_2-y_1$, that is 
\begin{equation} \label{2pol} 
    Z^{\rm OY}((y_1,2) \vert (y_2,1)) \overset{(d)}{=} Z^{\rm OY}[(0,2) \vert (y_2-y_1,1)].
\end{equation}
The law of $Z^{\rm OY}[(0,2) \vert (y_2-y_1,1)]$ is known very explicitly, its Laplace transform and its density are given 
respectively in Corollary 4.2 and in Theorem 5.1 of \cite{o2012directed}.
This allows to obtain the law of $Z^{\rm stat}_2(y_1,y_2)$. Moreover the multipoint 
correlations of $Z^{\rm OY}[(0,2) \vert (y,1)]$  are also known \cite[Section 5.2.2]{borodin2014macdonald}. 
At large distance (dilute limit), one has 
\be
\lim_{x\to\infty }\frac{1}{\sqrt{x}}\log Z^{\rm stat}_{2}(xy_1,xy_2)  \overset{(d)}{=} 
B_1(y_2) + B_1(y_1) + 
\max_{z \in [y_1,y_2]} \left\lbrace  B_2(z) - B_1(z) \right\rbrace ,
\label{eq:ell=2zerotemperature}
\ee 
where $B_1$ and $B_2$ are independent Brownian motions. 
As a process in $y_2$, for fixed $y_1$, it is also equal in law to 
\be 
B_1(y_1)+B_2(y_1) + \Lambda_1(y_2-y_1) 
\ee 
where $\Lambda_1(y)$ is independent of $B_{1,2}(y_1)$ and has the distribution of the largest eigenvalue of the $\mathrm{GUE}(2)$ Dyson Brownian motion (see \eqref{eq:zerotemperature} for a precise definition). This can be seen as the zero-temperature limit of \eqref{Z2stat}. 
 In Section \ref{sec:applications} below, we give an application of \eqref{Z2stat} where some explicit calculations are possible.

\begin{remark} The stochastic heat equation \eqref{eq:SHE} on $\mathbb R_+$ also admits stationary distributions, studied 
in \cite{barraquand2021steady} (see also \cite{barraquand2022stationary}), depending on two parameters, the boundary parameter $u$ and the drift at infinity $-v$. Denoting by $Z^{hs}$ the solution in the case of Dirichlet boundary conditions 
(which corresponds to $u \to +\infty$ in \cite{barraquand2021steady}) and droplet initial condition ($v=+\infty$)
$Z^{hs}(x,0\vert y,t=0)=\delta(y-x)$, the associated stationary measure for the polymer partition function  in half space is given by an O'Connell-Yor polymer partition function (this can be seen from Eq. (34) in \cite{barraquand2021steady}). In particular, using the relation \eqref{Z2stat} 
we obtain that 
\begin{equation}\label{singleDPhalfspace}
\lim_{t \to +\infty} \frac{Z^{hs}(x,-t \vert y,0)}{Z^{hs}(x,- t \vert z,0)} \overset{(d)}{=} \frac{Z_2^{\rm stat}(0,y)}{Z_2^{\rm stat}(0,z)}
\end{equation}
Hence there is a relation between the steady-state of the partition function of a single polymer in a half-space and the steady-state of the partition function of two 
non-intersecting polymers in full space. 
Moreover, for $v<0$, $Z_2^{\rm stat}(0,y; v, -v)$ (recall that $Z_{\ell}^{\rm stat}(\vec y; \vec a)$ was defined in Section \ref{sec:generalizations}) is also a stationary measure (with drift $-v$) for the polymer partition function in half space and Dirichlet boundary condition (see Eq. (35) in \cite{barraquand2021steady}). We conjecture that for $x=-vt$, 
\begin{equation}\label{singleDPhalfspacewithdrift}
\lim_{t \to +\infty} \frac{Z^{hs}(x,-t \vert y,0)}{Z^{hs}(x,- t \vert z,0)} \overset{(d)}{=} \frac{Z_2^{\rm stat}(0,y; v, -v)}{Z_2^{\rm stat}(0,z; v,-v)}.
\end{equation}
\end{remark}

\subsection{Midpoint} 

Consider $\ell$ very long non-crossing polymers and the positions of their midpoints $y_1, \dots y_\ell$ (see Figure \ref{fig:midpoint}).
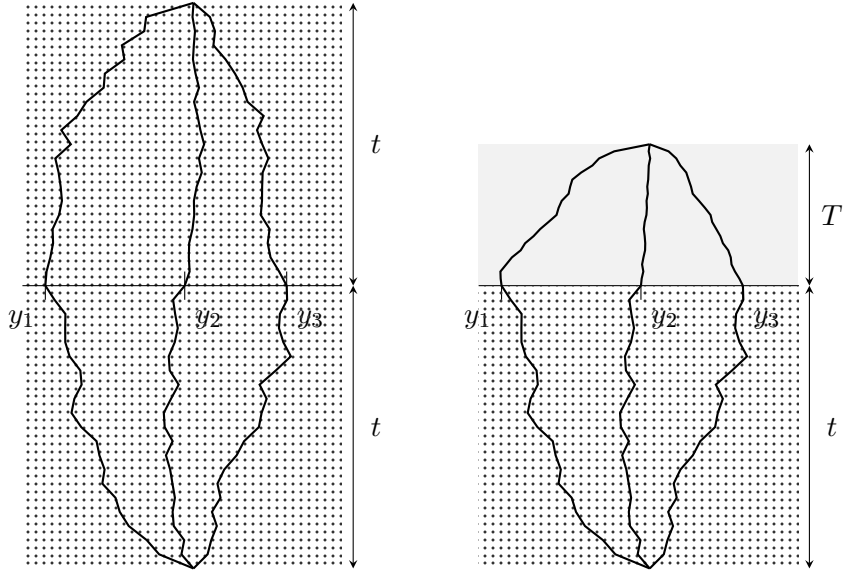
\begin{figure}
    \centering
     \begin{tikzpicture}[scale=1.5]
\begin{scope}[yscale=2.5] 
\draw (-1.2991,1.05) -- (-1.2991,0.95) node[anchor=north east]{$y_1$};
\draw (-0.0769828,1.05) -- (-0.0769828,0.95) node[anchor=north west]{$y_2$};
\draw (0.816921,1.05) -- (0.816921,0.95) node[anchor=north west]{$y_3$};

\fill[pattern=dots, pattern color=black!80] (-1.5,1) --(-1.5,0) -- (1.3,0) -- (1.3,1) -- cycle;
\fill[pattern=dots, pattern color=black!80] (-1.5,1) --(-1.5,2) -- (1.3,2) -- (1.3,1) -- cycle;

\draw[stealth-stealth] (1.4,1) -- (1.4,2);
\draw (1.6,1.5) node{$t$};

\draw[stealth-stealth] (1.4,1) -- (1.4,0);
\draw (1.6,0.5) node{$t$};

\draw (-1.5,1) -- (1.3,1);

\draw[thick] ((0.,0.) -- 
(-0.30143,0.05) -- 
(-0.409552,0.1) -- 
(-0.567782,0.15) -- 
(-0.646555,0.2) -- 
(-0.684085,0.25) -- 
(-0.799482,0.3) -- 
(-0.779747,0.35) -- 
(-0.825065,0.4) -- 
(-0.850036,0.45) -- 
(-0.984687,0.5) -- 
(-1.07001,0.55) -- 
(-1.04471,0.6) -- 
(-0.981216,0.65) -- 
(-0.992227,0.7) -- 
(-1.08658,0.75) -- 
(-1.12445,0.8) -- 
(-1.12515,0.85) -- 
(-1.12334,0.9) -- 
(-1.21874,0.95) -- 
(-1.2991,1.);
\draw[thick] (0,0.) -- 
(-0.105995,0.05) -- 
(-0.0813379,0.1) -- 
(-0.172984,0.15) -- 
(-0.183264,0.2) -- 
(-0.163246,0.25) -- 
(-0.184516,0.3) -- 
(-0.205019,0.35) -- 
(-0.241569,0.4) -- 
(-0.182845,0.45) -- 
(-0.252675,0.5) -- 
(-0.26036,0.55) -- 
(-0.19732,0.6) -- 
(-0.129879,0.65) -- 
(-0.208485,0.7) -- 
(-0.217275,0.75) -- 
(-0.164383,0.8) -- 
(-0.137803,0.85) -- 
(-0.160996,0.9) -- 
(-0.178837,0.95) -- 
(-0.0769828,1.);
\draw[thick] (0.,0.) -- 
(0.123397,0.05) -- 
(0.155497,0.1) -- 
(0.204339,0.15) -- 
(0.161931,0.2) -- 
(0.23662,0.25) -- 
(0.209985,0.3) -- 
(0.269903,0.35) -- 
(0.378128,0.4) -- 
(0.443318,0.45) -- 
(0.573329,0.5) -- 
(0.596939,0.55) -- 
(0.641084,0.6) -- 
(0.57959,0.65) -- 
(0.721362,0.7) -- 
(0.85343,0.75) -- 
(0.787414,0.8) -- 
(0.749207,0.85) -- 
(0.757518,0.9) -- 
(0.820677,0.95) -- 
(0.816921,1.);

\draw[thick] (0.,2.) -- 
(-0.415484,1.95) -- 
(-0.431754,1.9) -- 
(-0.627879,1.85) -- 
(-0.606328,1.8) -- 
(-0.777359,1.75) -- 
(-0.787515,1.7) -- 
(-0.939009,1.65) -- 
(-1.01841,1.6) -- 
(-1.15911,1.55) -- 
(-1.07556,1.5) -- 
(-1.2118,1.45) -- 
(-1.18277,1.4) -- 
(-1.16474,1.35) -- 
(-1.15354,1.3) -- 
(-1.17879,1.25) -- 
(-1.23552,1.2) -- 
(-1.2312,1.15) -- 
(-1.25531,1.1) -- 
(-1.28855,1.05) -- 
(-1.2991,1.);
\draw[thick]  (0.,2.) -- 
(-0.00865224,1.95) -- 
(0.0131024,1.9) -- 
(0.00270302,1.85) -- 
(-0.0047464,1.8) -- 
(0.0176092,1.75) -- 
(0.0482578,1.7) -- 
(0.00467315,1.65) -- 
(0.0357554,1.6) -- 
(0.0602152,1.55) -- 
(0.0930976,1.5) -- 
(0.0478634,1.45) -- 
(0.0724006,1.4) -- 
(0.0276482,1.35) -- 
(0.00359049,1.3) -- 
(0.00624363,1.25) -- 
(-0.00894476,1.2) -- 
(-0.0381336,1.15) -- 
(-0.0409067,1.1) -- 
(-0.0342688,1.05) -- 
(-0.0769828,1.);
\draw[thick]  (0.,2.) -- 
(0.143668,1.95) -- 
(0.18849,1.9) -- 
(0.174609,1.85) -- 
(0.281322,1.8) -- 
(0.363575,1.75) -- 
(0.431096,1.7) -- 
(0.473896,1.65) -- 
(0.615144,1.6) -- 
(0.556915,1.55) -- 
(0.602463,1.5) -- 
(0.656343,1.45) -- 
(0.606201,1.4) -- 
(0.601749,1.35) -- 
(0.600422,1.3) -- 
(0.58968,1.25) -- 
(0.667141,1.2) -- 
(0.650518,1.15) -- 
(0.675242,1.1) -- 
(0.750579,1.05) -- 
(0.816921,1.);
\end{scope}

\begin{scope}[xshift=4cm, yscale=2.5]

\draw (-1.2991,1.05) -- (-1.2991,0.95) node[anchor=north east]{$y_1$};
\draw (-0.0769828,1.05) -- (-0.0769828,0.95) node[anchor=north west]{$y_2$};
\draw (0.816921,1.05) -- (0.816921,0.95) node[anchor=north west]{$y_3$};

\fill[pattern=dots, pattern color=black!80] (-1.5,1) --(-1.5,0) -- (1.3,0) -- (1.3,1) -- cycle;
\fill[black!05!] (-1.5,1) --(-1.5,1.5) -- (1.3,1.5) -- (1.3,1) -- cycle;

\draw (-1.5,1) -- (1.3,1);
\draw[stealth-stealth] (1.4,1) -- (1.4,1.5);
\draw (1.6,1.25) node{$T$};

\draw[stealth-stealth] (1.4,1) -- (1.4,0);
\draw (1.6,0.5) node{$t$};

\draw[thick] ((0.,0.) -- 
(-0.30143,0.05) -- 
(-0.409552,0.1) -- 
(-0.567782,0.15) -- 
(-0.646555,0.2) -- 
(-0.684085,0.25) -- 
(-0.799482,0.3) -- 
(-0.779747,0.35) -- 
(-0.825065,0.4) -- 
(-0.850036,0.45) -- 
(-0.984687,0.5) -- 
(-1.07001,0.55) -- 
(-1.04471,0.6) -- 
(-0.981216,0.65) -- 
(-0.992227,0.7) -- 
(-1.08658,0.75) -- 
(-1.12445,0.8) -- 
(-1.12515,0.85) -- 
(-1.12334,0.9) -- 
(-1.21874,0.95) -- 
(-1.2991,1.);
\draw[thick] (0,0.) -- 
(-0.105995,0.05) -- 
(-0.0813379,0.1) -- 
(-0.172984,0.15) -- 
(-0.183264,0.2) -- 
(-0.163246,0.25) -- 
(-0.184516,0.3) -- 
(-0.205019,0.35) -- 
(-0.241569,0.4) -- 
(-0.182845,0.45) -- 
(-0.252675,0.5) -- 
(-0.26036,0.55) -- 
(-0.19732,0.6) -- 
(-0.129879,0.65) -- 
(-0.208485,0.7) -- 
(-0.217275,0.75) -- 
(-0.164383,0.8) -- 
(-0.137803,0.85) -- 
(-0.160996,0.9) -- 
(-0.178837,0.95) -- 
(-0.0769828,1.);
\draw[thick] (0.,0.) -- 
(0.123397,0.05) -- 
(0.155497,0.1) -- 
(0.204339,0.15) -- 
(0.161931,0.2) -- 
(0.23662,0.25) -- 
(0.209985,0.3) -- 
(0.269903,0.35) -- 
(0.378128,0.4) -- 
(0.443318,0.45) -- 
(0.573329,0.5) -- 
(0.596939,0.55) -- 
(0.641084,0.6) -- 
(0.57959,0.65) -- 
(0.721362,0.7) -- 
(0.85343,0.75) -- 
(0.787414,0.8) -- 
(0.749207,0.85) -- 
(0.757518,0.9) -- 
(0.820677,0.95) -- 
(0.816921,1.);

\draw[thick] (0.,1.5) -- 
(-0.311098,1.475) -- 
(-0.443875,1.45) -- 
(-0.502154,1.425) -- 
(-0.601156,1.4) -- 
(-0.678613,1.375) -- 
(-0.697669,1.35) -- 
(-0.711061,1.325) -- 
(-0.771613,1.3) -- 
(-0.799025,1.275) -- 
(-0.817052,1.25) -- 
(-0.887929,1.225) -- 
(-0.960719,1.2) -- 
(-1.01865,1.175) -- 
(-1.09263,1.15) -- 
(-1.13115,1.125) -- 
(-1.20828,1.1) -- 
(-1.25831,1.075) -- 
(-1.30908,1.05) -- 
(-1.30695,1.025) -- 
(-1.2991,1.);
\draw[thick]  (0.,1.5) -- 
(-0.00897148,1.475) -- 
(0.00391553,1.45) -- 
(-0.00861495,1.425) -- 
(-0.0108941,1.4) -- 
(-0.0133113,1.375) -- 
(-0.021714,1.35) -- 
(-0.0165484,1.325) -- 
(-0.0278833,1.3) -- 
(-0.0247638,1.275) -- 
(-0.0289309,1.25) -- 
(-0.0443874,1.225) -- 
(-0.0520282,1.2) -- 
(-0.0441431,1.175) -- 
(-0.0393319,1.15) -- 
(-0.0407934,1.125) -- 
(-0.0549951,1.1) -- 
(-0.0669207,1.075) -- 
(-0.0561042,1.05) -- 
(-0.0706213,1.025) -- 
(-0.0769828,1.);
\draw[thick]  (0.,1.5) -- 
(0.165031,1.475) -- 
(0.232699,1.45) -- 
(0.26819,1.425) -- 
(0.320383,1.4) -- 
(0.340936,1.375) -- 
(0.362046,1.35) -- 
(0.376879,1.325) -- 
(0.43171,1.3) -- 
(0.484564,1.275) -- 
(0.533122,1.25) -- 
(0.54767,1.225) -- 
(0.583816,1.2) -- 
(0.619365,1.175) -- 
(0.641514,1.15) -- 
(0.688169,1.125) -- 
(0.710399,1.1) -- 
(0.717164,1.075) -- 
(0.74725,1.05) -- 
(0.785786,1.025) -- 
(0.816921,1.);

\end{scope}

\end{tikzpicture}
\caption{Left: Midpoint distribution of $\ell$ non-intersecting polymers of length $2t$. Right: Endpoint distribution of $\ell$ non-intersecting polymers in a confining potential. It can also be seen as the midpoint distribution for polymers which in the lower (dotted) region are non-intersecting and
in a white noise environment, while in the upper (shaded) region there is no noise and they are simply Brownian paths 
 (with or without the non-crossing constraint)}
    \label{fig:midpoint}
\end{figure}
The stationary partition sum for the midpoints is simply the product
\be \label{eq:ZstatMidpoint}
Z_\ell^{\rm midpoint}(\vec y) := Z^{\rm stat}_\ell(\vec y) \tilde Z^{\rm stat}_\ell(\vec y)
\ee 
where $Z^{\rm stat}_\ell(\vec y)$ is given by formula \eqref{eq:defZstat} and $\tilde Z^{\rm stat}_\ell(\vec y)$ is given by the same
formula with the set of $B_j(x)$ are replaced by an independent set of Brownian motions $\tilde B_j(x)$. A graphical 
description is given for $\ell=3$ in Fig. \ref{fig:midpoint2} (left).

\begin{figure}
	\centering
\begin{center}	
	\begin{tikzpicture}[scale=0.9]
	\begin{scope}
\draw[gray] (0,0) -- (7,0) node[anchor=west] {$B_3$};
\draw[gray] (0,1) -- (7,1) node[anchor=west] {$B_2$};
\draw[gray] (0,2) -- (7,2) node[anchor=west] {$B_1$};
\draw[gray] (0,3) -- (7,3) node[anchor=west] {$\tilde B_1$};
\draw[gray] (0,4) -- (7,4) node[anchor=west] {$\tilde B_2$};
\draw[gray] (0,5) -- (7,5) node[anchor=west] {$\tilde B_3$};
\draw (4,2) node[anchor=north] {$y_2$};
\draw (6,2) node[anchor=north] {$y_3$};
\draw (1,2) node[anchor=north] {$y_1$};
\draw[->] (3,0) -- (3,0.5);
\draw (3,0) -- (3,1);
\draw[->] (2,1) -- (2,1.5);
\draw (2,1) -- (2,2);
\draw[->] (5,1) -- (5,1.5);
\draw (5,1) -- (5,2);
\draw[->] (4,5) -- (4,4.5);
\draw (4,5) -- (4,4);
\draw[->] (3,4) -- (3,3.5);
\draw (3,4) -- (3,3);
\draw[->] (5,4) -- (5,3.5);
\draw (5,4) -- (5,3);
\draw (1,2) -- (1,3);
\draw (4,2) -- (4,3);
\draw (6,2) -- (6,3);
\draw[ultra thick] (0,0) -- (3,0);
\draw[ultra thick] (0,1) -- (2,1);
\draw[ultra thick] (3,1) -- (5,1);
\draw[ultra thick] (0,2) -- (1,2);
\draw[ultra thick] (2,2) -- (4,2);
\draw[ultra thick] (5,2) -- (6,2);
\draw[ultra thick] (0,3) -- (1,3);
\draw[ultra thick] (3,3) -- (4,3);
\draw[ultra thick] (5,3) -- (6,3);
\draw[ultra thick] (0,4) -- (3,4);
\draw[ultra thick] (4,4) -- (5,4);
\draw[ultra thick] (0,5) -- (4,5);
\draw (0,-0.1) node[anchor=north] {$0$} -- (0,0.1);
\draw (0,0.9) node[anchor=north] {$0$} -- (0,1.1);
\draw (0,1.9) node[anchor=north] {$0$} -- (0,2.1);
\draw (0,2.9) node[anchor=north] {$0$} -- (0,3.1);
\draw (0,3.9) node[anchor=north] {$0$} -- (0,4.1);
\draw (0,4.9) node[anchor=north] {$0$} -- (0,5.1);
	\end{scope}
	\begin{scope}[xshift=8.5cm]
\draw[gray] (0,0) -- (7,0) node[anchor=west] {$B_3$};
\draw[gray] (0,1) -- (7,1) node[anchor=west] {$B_2$};
\draw[gray] (0,2) -- (7,2) node[anchor=west] {$B_1$};
\draw[gray] (0,3) -- (7,3) node[anchor=west] {$\tilde B_1$};
\draw[gray] (0,4) -- (7,4) node[anchor=west] {$\tilde B_2$};
\draw[gray] (0,5) -- (7,5) node[anchor=west] {$\tilde B_3$};
\draw (4,2) node[anchor=north] {$y_2$};
\draw (6,2) node[anchor=north] {$y_3$};
\draw (1,2) node[anchor=north] {$y_1$};
\draw[->] (3,0) -- (3,0.5);
\draw (3,0) -- (3,1);
\draw[->] (2,1) -- (2,1.5);
\draw (2,1) -- (2,2);
\draw[->] (5,1) -- (5,1.5);
\draw (5,1) -- (5,2);
\draw (1,2) -- (1,3);
\draw (4,2) -- (4,3);
\draw (6,2) -- (6,3);
\draw[->] (3,3) -- (3,3.5);
\draw (3,3) -- (3,4);
\draw[->] (5,3) -- (5,3.5);
\draw (5,3) -- (5,4);
\draw[->] (4,4) -- (4,4.5);
\draw (4,4) -- (4,5);
\draw[ultra thick] (0,0) -- (3,0);
\draw[ultra thick] (0,1) -- (2,1);
\draw[ultra thick] (3,1) -- (5,1);
\draw[ultra thick] (0,2) -- (1,2);
\draw[ultra thick] (2,2) -- (4,2);
\draw[ultra thick] (5,2) -- (6,2);
\draw[ultra thick] (1,3) -- (3,3);
\draw[ultra thick] (4,3) -- (5,3);
\draw[ultra thick] (6,3) -- (7,3);
\draw[ultra thick] (3,4) -- (4,4);
\draw[ultra thick] (5,4) -- (7,4);
\draw[ultra thick] (4,5) -- (7,5);

\draw (0,-0.1) node[anchor=north] {$0$} -- (0,0.1);
\draw (0,0.9) node[anchor=north] {$0$} -- (0,1.1);
\draw (0,1.9) node[anchor=north] {$0$} -- (0,2.1);
\draw (7,2.9)  -- (7,3.1) node[anchor=south] {$x$};
\draw (7,3.9)  -- (7,4.1) node[anchor=south] {$x$};
\draw (7,4.9)  -- (7,5.1) node[anchor=south] {$x$};
	\end{scope}
	\end{tikzpicture}
	\end{center}
	\caption{Left: Graphical interpretation of the partition function  $Z^{\rm midpoint}$ defined in \eqref{eq:ZstatMidpoint}, which is proportional to the probability density of the midpoints of $\ell$ non-intersecting polymers.
	Right: Graphical interpretation of the partition function $\tilde Z^{\rm midpoint}$ defined in \eqref{eq:ZstatMidpoint2}, which is proportional to the probability density of the midpoints of $\ell$ non-intersecting polymers.}
	\label{fig:midpoint2}
\end{figure}
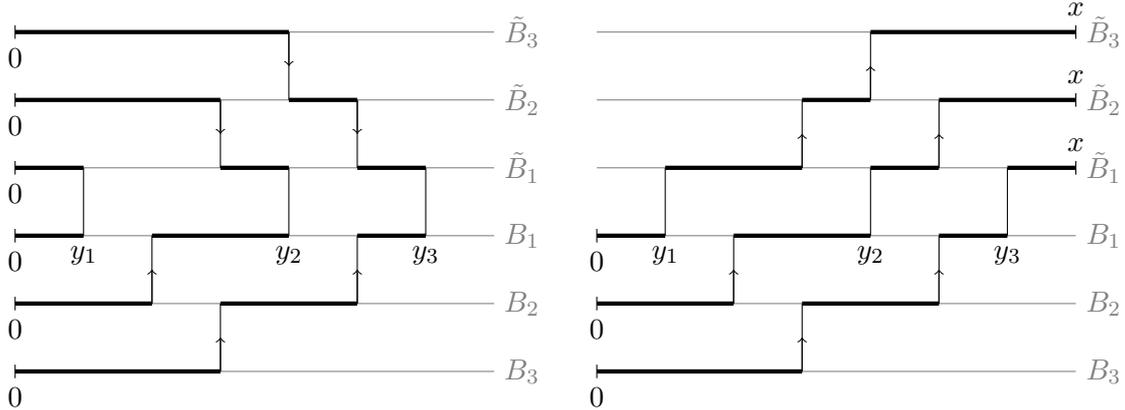

More precisely
one has
\begin{equation} 
\lim_{t \to +\infty} 
 \frac{ Z_{\ell}(\vec x;-t\vert \vec y ;0) Z_{\ell}(\vec y;0\vert \vec x' ;t) }{Z_{\ell}(\vec x;-t\vert \vec z ;0) Z_{\ell}(\vec z;0\vert \vec x' ;t) } \overset{(d)}{=}
\frac{Z^{\rm stat}_\ell(\vec y)}{Z^{\rm stat}_\ell(\vec z)}  \frac{\tilde Z^{\rm stat}_\ell(\vec y)}{\tilde Z^{\rm stat}_\ell(\vec z)}, 
\label{eq:midpoint}
\end{equation}
Note that the ratio $\frac{\tilde Z^{\rm stat}_\ell(\vec y)}{\tilde Z^{\rm stat}_\ell(\vec z)}$ has the same law as $\frac{\tilde Z^{\rm stat}_\ell(x-y_\ell, \dots, x-y_1)}{\tilde Z^{\rm stat}_\ell(x-z_\ell, \dots, x-z_1)}$, where $x>y_\ell$. Thus, the limit \eqref{eq:midpoint} can also be described by the ratios of the partition function 
\begin{equation} \label{eq:ZstatMidpoint2}
\tilde Z_\ell^{\rm midpoint}(\vec y) := Z^{\rm stat}_\ell(y_1, \dots, y_\ell) \tilde Z^{\rm stat}_\ell(x-y_\ell, \dots, x-y_1).
\end{equation}
This partition function has another graphical interpretation, see Fig. \ref{fig:midpoint2} (right), where $\vec y$ represents the positions of the midpoints of $\ell$ non intersecting semi-discrete polymers. We can view $\tilde Z_\ell^{\rm midpoint}(\vec y) $ as proportional to the probability density of the midpoints  for $\ell$ very long non-intersecting polymers. The normalization of this probability measure (restricting $\vec y\in [0,x]^\ell$) is given by 
\begin{equation}
    \int_{\mathbb W_\ell\cap [0,x]^\ell} Z_\ell^{\rm midpoint 2}(\vec y)d\vec y = \det\left(Z^{\rm OY}[(0,i)\vert (x,j-\ell)]\right)_{i,j=1}^\ell,
\end{equation}
where we recall that $\mathbb W_\ell$ was defined in \eqref{eq:defWeyl} and  $Z^{\rm OY}$ is the partition function of a single semi-discrete polymer defined in \eqref{eq:defZOY} (the lines of the O'Connell-Yor polymer are here indexed by integers $\lbrace 1-\ell, 2-\ell, \dots, \ell-1,\ell \rbrace$).

\subsection{Applications to normalized endpoint distributions}
\label{sec:applications}
As is the case for $\ell=1$, the partition sum $Z_\ell^{\rm stat}(\vec y)$ grows unboundedly at large $\vec y$ and 
cannot be normalized on the whole line $\mathbb{R}^\ell$. To define normalized stationary endpoint probabilities, i.e.
for very long non-crossing polymers one could consider the system in a restricted geometry (such as a circle, or in presence of
walls) but then the stationary measure is different and dependent on the boundaries. 
There are still however some interesting observables related to the endpoint probabilities of $\ell$ non crossing
polymers which one can define from our present result, leaving their detailed study for the future.

\subsection*{Conditional probability in a fixed interval} One can study the endpoint probability conditioned to the event that all $\ell$ polymers belong to some fixed interval $[0,L]$. It reads 
\begin{equation}
\mathbb{P} \left( \vec y\in E \vert \vec y\in [0,L] \right) = \frac{\int_{\mathbb W_\ell\cap E^\ell} \mathrm d\vec y Z_{\ell}^{\rm stat}(\vec y)}{\int_{\mathbb W_\ell \cap [0,L]^\ell} \mathrm d\vec y Z_{\ell}^{\rm stat}(\vec y)}.
\end{equation}

\subsection*{Polymers with endpoints submitted to a potential} It is possible to consider a model where the polymer endpoints are submitted to an additional potential,
that we denote $V(\vec y)$. It can result from interactions between the endpoints, or from an external potential acting on
each endpoint $V(\vec y)= \sum_i v(y_i)$. An experimental example of the latter is a tunnelling microscope tip acting on the endpoint of vortex lines
\cite{ge2016nanoscale}.
If $V(\vec y)$ is globally confining, with associated equilibrium Gibbs measure $e^{- V(y)}$, 
it leads to a normalizable endpoint probability for $\ell$ non-crossing very long polymers
\begin{equation}
\mathbb{P} \left( \vec y \right) = \frac{Z^{\rm stat}_{\ell}(\vec y) e^{- V(\vec y)}}{\int_{\mathbb W_\ell} d\vec z Z_{\ell}^{\rm stat}(\vec z) e^{- V(\vec z)} }
\label{eq:endpointpotential}
\end{equation}
We assume here that $V(\vec y)$ grows faster than $\sqrt{|y|}$. 

One realization which does not involve explicitly an additional potential, is to consider the midpoint of $\ell$ non-crossing 
directed polymers in the case where there is no disorder on the top part (assumed here to be of length $T$), 
and all points equal (within $\epsilon$) at the final time (see Figure \ref{fig:midpoint}). 
In that case 
\be 
e^{- V(\vec y)} \to \Delta^\beta(\vec y) e^{- \vec y^2/(4 T)}
\ee 
 where $\beta=1,0$ depending whether one enforces or not the non-crossing constraint in the top part. The 
Vandermonde determinant $\Delta(\vec y)$ results from the non-crossing constraint. Note that now the normalization integral is convergent and finite.

\subsection*{Two non-crossing polymers in an atypical direction}

Consider again the partition sum of two non-crossing polymers $Z^{\rm stat}_\ell(y_1,y_2; a,a)$ with drifts $a_i=- b$ with $b>0$.
We are interested in the PDF, $P(y)$, of the distance $y=y_2-y_1$ between the two endpoints, in the stationary state. 
It can be written as 
\be 
P(y) = \frac{Z(y)}{\int_0^{+\infty}  dy Z(y)}  \quad , \quad Z(y):= Z^{\rm stat}_\ell(0,y; -b,-b) = Z^{\rm OY}((0,2) \vert (y,1))
\ee 
Note that the dependence in $y_1$ alone cancels (i.e. $y_1$ can be set to zero). In the second formula we have used
\eqref{2pol}, where now in the OY partition sum the two Brownians have drifts $-b$. As in \eqref{eq:mainresultwithdrifts}
$P(y)$ is also the PDF of the distance at time zero of the endpoints of two very long polymers starting both at $- b t$ at time $-t$.
It is equivalent to conditioning the two polymers such that the first one (hence both) ends in an atypical direction
with space time slope $b=x/t$. The nice observation is that upon this conditioning the partition sum $Z(y)$ can be normalized,
i.e. the second polymer is then bound to the first one. 

We can compute the thermal cumulants of $y$ as follows. Let us define
\be 
{\sf Z}(\alpha) =  \int_0^{+\infty} dy Z(y) e^{- \alpha y} 
\ee 
Then the thermal cumulants are given as $\langle y^p \rangle^c= (-1)^p \partial_\alpha^p \log {\sf Z}(\alpha)|_{\alpha=0}$. It turns out 
that the full PDF of ${\sf Z}(\alpha)$ can be obtained from known formula for the OY point to line problem. Using
Theorem 3 in \cite{fitzgerald2020point} (setting $n=2$ there) one finds that
\be 
{\sf Z}(\alpha) =  2 W_{11} (W_{12}+W_{21})
\ee 
where the $W_{11},W_{12},W_{21}$ are three i.i.d inverse gamma random variables $\mathrm{Gamma}^{-1}(2 \alpha + 2 b)$ with rate unity. Using that $\mathrm{Gamma}(c_1,1)+\mathrm{Gamma}(c_2,1)\overset{(d)}{=}\mathrm{Gamma}(c_1+c_2, 1)$, and \break $\mathbb E[\log \mathrm{Gamma}(c,1)]=\Psi(c)$  (where $\Psi(z)=\partial_z \log \Gamma(z)$ is the digamma function) one finds
\be 
\mathbb E[\log {\sf Z}(\alpha) ]= \log 2 + \Psi(4 \alpha + 4 b) - 3 \Psi(2 \alpha + 2 b) 
\ee 
Hence one obtains the disorder averaged cumulants
\be 
\mathbb E\left[{\langle y^p \rangle^c}\right] = (-2)^p (2^p \Psi_p(4 b)- 3 \Psi_p(2 b)),
\ee 
where $\Psi_p(z)=\partial_z^p \Psi(z)$ is the $p$-th derivative of the digamma function. 
In the limit of small $b \ll 1$ one finds
\be 
\mathbb E\left[\langle y \rangle\right] \simeq \frac{5}{4 b^2} \quad , \quad \mathbb E\left[\langle y^2 \rangle- \langle y \rangle^2\right] \simeq \frac{5}{2 b^3} \label{thermal} 
\ee 
and we see that the effect of thermal fluctuations, measured for instance by the ratio \break 
$\mathbb E\left[{\langle y^2 \rangle- \langle y \rangle^2}\right]/\left(\mathbb E\left[{\langle y \rangle}\right]\right)^2$,
become subdominant as $b \to 0$. This is expected since it corresponds to the large scale limit. 
Note that these calculations are very similar to the ones performed in \cite{barraquand2021kardar}
for the problem of a single polymer bound to a wall with wall parameter $u=A+\frac{1}{2}=-b<0$. Once again we find that
the two problems are close cousins.

Note that the same problem where the two polymers come from different directions, i.e. have different drifts $(a_1,a_2)=(-b_1,-b_2)$ with $b_1>b_2$ can be treated similarly, with $W_{11}=\Gamma^{-1}(2 \alpha + b_1+b_2)$, $W_{21}=\Gamma^{-1}(2 \alpha + b_1)$, $W_{12}=\Gamma^{-1}(2 \alpha + b_2)$.

\subsubsection*{Limit $b \to 0$ and Dyson's Brownian motion (DBM)} In the limit $b \to 0$ one can rescale the distances as $y=\tilde y/b^2$. Then the variable $\tilde y$
is determined by 
\be 
\tilde y = {\rm argmax}_{z \in \mathbb{R}_+} ( \Lambda_1(z) - z) 
\ee 
where $\Lambda_1(z)$ is the largest eigenvalue in a $\mathrm{GUE}(2)$ DBM. The explicit formula for the 
PDF $p(t)$ of $t=\tilde y$ is obtained in the Appendix, see Eq. \eqref{pt}. From this formula one
obtains for instance the behavior of the lowest moments in the limit $b \ll 1$ as
\be 
\mathbb E\left[\langle y \rangle\right] \simeq \frac{5}{4 b^2} \quad , \quad \mathbb E\left[\langle y^2 \rangle \right] 
\simeq \frac{29}{8 b^4}
\ee 
where the result for the first moment coincides with the result in \eqref{thermal} although it was
obtained by a completely different method. One sees that the second moment diverges much faster
than the thermal one at small
$b$ ($O(1/b^4)$ instead as $O(1/b^3)$) consistent with the above discussion. 
\\

\subsection{Method} 
\label{sec:method}
The partition function $Z_{\ell}(\vec x,0\vert \vec y,t)$ defined in \eqref{eq:defZl} vanishes as soon as some of the $x_i$ are equal or some of the $y_i$ are equal. Let us define 
\begin{equation}
M_\ell(t, \vec x, \vec y)  = \frac{\det\left(Z_1(x_i,0\vert y_j,t)\right)_{i,j=1}^\ell}{\Delta(\vec x)\Delta(\vec y)},
\end{equation}
extended by continuity when some $x_i$ or $y_i$ are equal (see \cite{o2016multi} about mathematical subtleties related to this continuous extension).   When all coordinates are equal, it can be proven \cite{o2016multi},  at least for a regularized noise, that 
\begin{equation}
M_\ell(t, x\vec 1, y\vec 1) = c_n^2 \det\left( \partial_x^{i-1}\partial_y^{j-1} Z(x,0\vert y,t) \right)_{i,j=1}^{\ell} 
\label{eq:relationMndeterminant}
\end{equation}
where $c_n=1/\prod_{j=1}^{\ell-1}j!$, and in that case, $M_\ell(t, x\vec 1, y\vec 1) $ can be identified with the partition function for $\ell$ non-intersecting continuous directed polymer paths all starting from $x$ and all ending at  $y$. The collection of random processes $(M_1(t, x\vec 1, y\vec 1), \dots,  M_\ell(t, x\vec 1, y\vec 1))$ is called the O'Connell-Warren multilayer stochastic heat equation in the mathematical literature. It was rigorously constructed in \cite{o2016multi} through a chaos series expansion.
Let us further define the ratios 
\begin{equation}
    R_{\ell} (t,x,y) = \frac{M_{\ell}(t, x\vec 1, y\vec 1)}{M_{\ell-1}(t, x\vec 1, y\vec 1)}
    \label{eq:defratiosRell}
\end{equation}
with the convention that $M_0=1$.

While it is obvious that for a fixed $\vec x$, the set of all non-intersecting partition functions $\lbrace M_\ell(t, \vec x, \vec y)\rbrace_{\vec y\in \mathbb W_{\ell}}$ evolves in a Markovian way as $t$ increases,  \cite{o2016multi} discovered that for any fixed $\ell$, the time evolution of the much smaller set of partition functions  
\begin{equation}
     \left\lbrace R_1(t, x ,y), R_2(t, x, y), \dots, R_{\ell}(t, x ,y) \right\rbrace_{y\in \mathbb R},
     \label{eq:Markovianratios}
\end{equation}
is also Markovian. In other terms, the partition functions for non-intersecting polymers ending at the same point behave in a Markovian way, which is very nontrivial. This property comes from an analogous result in the discrete setting, where the dynamics of the discrete analogue of \eqref{eq:Markovianratios} are known as the geometric RSK algorithm, and related to the log-gamma discrete directed polymer model \cite{corwin2014tropical}, see Section \ref{sec:geometricRSK} below. These discrete Markovian dynamics are perfectly explicit and a family of invariant measures is described in \cite{corwin2014tropical}. In Section \ref{sec:discretecase}, we interpret these results in terms of partition functions for non-intersecting log-gamma directed polymers. In particular, we define a variant of the log-gamma directed polymer model such that ratios of partition functions for non-intersecting polymers are stationary, i.e. their distribution is invariant with respect to shifting the terminal points. Going to the continuous limit, as we explain in Section \ref{sec:KPZlimit}, we deduce that the process \eqref{eq:Markovianratios} admits a family of invariant measures given by the exponentials of $\ell$ independent standard Brownian  motions with equal drift (see more details in Section \ref{sec:invariantOConnellWarren}). Then, we deduce our main result \eqref{eq:multipolymers} from two different methods:
\begin{enumerate}
    \item Using an explicit relation obtained in \cite{o2016multi} expressing partition functions for non-intersecting continuous polymers with distinct endpoints $M_{\ell}(t, x\vec 1, \vec y)$ in terms of partition functions for non-intersecting polymers with coinciding endpoints, that is in terms of the $R_{j}(x,0\vert y,t)$ for $y\in \mathbb R$ and $j\leq \ell$. 
    \item Alternatively, we show that we may also obtain a discrete analogue of \eqref{eq:multipolymers}  (see Proposition \ref{prop:log-gammaarbitrarymultipoint}). Going to the continuous limit (see Section \ref{sec:alternativemethod}), this allows to determine the invariant measures of the Markov process $t\mapsto \lbrace M_\ell(t, \vec x, \vec y)\rbrace_{\vec y\in \mathbb W_{\ell}}$ from which we deduce  \eqref{eq:multipolymers}. 
\end{enumerate}

\begin{remark}
For fixed $t$, the process \eqref{eq:Markovianratios} has the same distribution as the so-called  KPZ line ensemble introduced in \cite{corwin2016kpz}. Indeed, \cite{nica2021intermediate} proved the convergence of partition functions of non-intersecting semi-discrete  O'Connell-Yor polymers to the multilayer O'Connell-Warren process, which allowed to prove that the O'Connell-Warren multilayer SHE at a given time is given by the KPZ line ensemble. The invariant measure of the O'Connell-Warren process that we find in this paper is consistent with the Gibbs property of the KPZ line ensemble (a specific Gibbs resampling property which implies local Brownianity of the stochastic processes \eqref{eq:Markovianratios}). In particular, the prediction from Remark \ref{rem:t2/3scale} could also be deduced from the convergence of the KPZ line ensemble to the Airy line ensemble \cite{corwin2014brownian} at large scale and the locally Brownian nature of the Airy line ensemble (which is a consequence of the Brownian Gibbs property satisfied by the Airy line ensemble).
\end{remark}

\subsection*{Acknowledgments} G.B. was partially supported by ANR grant ANR-21-CE40-0019.  P.L.D. acknowledges support from the ANR grant ANR-17-CE30-0027-01 RaMaTraF. This article is based upon work supported by the National Science Foundation under Grant No. DMS-1928930 while the authors participated in a program hosted by the Mathematical Sciences Research Institute in Berkeley, California, during the Fall 2021 semester. We also acknowledge hospitality and support from  Galileo Galilei Institute,
during the participation of both authors to the  scientific program on    ``Randomness, Integrability, and 
Universality'' in the Spring 2022.

\section{A stationary model for non-intersecting log-gamma polymers}
\label{sec:discretecase}
The aim of this Section is to obtain discrete analogues of \eqref{eq:multipolymers} and \eqref{eq:defZstat}. 
\subsection{Preliminary notations}
\label{sec:notations}
For a discrete polymer model on $\Z^2$ with Boltzmann weights $w_{i,j}$ we define the partition function as usual by 
\begin{equation}
    \ZZ_1[\mathbf x\vert \mathbf y] = \sum_{\pi: \mathbf x \to\mathbf y} \prod_{(i,j)\in \pi} w_{i,j}
\end{equation}
where the sum runs over upright paths $\pi$ in $\mathbb Z^2$ from $\mathbf x$ to $\mathbf y$.
We also define the partition function
\begin{equation}
    \ZZ_{\ell}[\xx_1, \dots,\xx_{\ell}\vert \yy_1, \dots, \yy_{\ell}]= \sum_{\pi_1, \dots, \pi_{\ell}} \prod_{k=1}^{\ell} \prod_{(i,j)\in \pi_k} w_{i,j}, 
    \label{eq:defpartitionfunctiongeneral}
\end{equation} 
where the sum runs over  $\ell$-tuples of non-intersecting polymers paths $\pi_i$ ($1\le i\le \ell$)  in $\Z^2$  joining the points $\xx_i$ to $\yy_i$. 
By the Lindstr\"om-Gessel-Viennot/Karlin-McGregor lemma,  $\ZZ_\ell$ can be computed as a determinant of single polymer partition functions: 
\begin{equation}
\ZZ_{\ell}[\xx_1, \dots,\xx_{\ell}\vert \yy_1, \dots, \yy_{\ell}] = \det\left(\ZZ_1[\xx_i\vert\yy_j]\right)_{i,j=1}^{\ell}.
\label{eq:LGVdiscrete}
\end{equation}

\subsection{The geometric RSK correspondence and  polymer partition functions}
\label{sec:geometricRSK}
In this Section we will discuss the evolution of  the partition function for $\ell$ non-intersecting polymers, starting close to $(1,1)$ and arriving close to $(n,m)$,  as the point $(n,m)$ varies. It will be convenient to use the shorthand notation
\begin{equation}
\Zcal_{\ell}(n,m) = \sum_{\pi} \prod_{r=1}^{\ell}\prod_{(i,j)\in \pi_r} w_{i,j}. 
\label{eq:partitionfunctionell}
\end{equation}
to be the partition function where the sum runs over $\ell$-tuples $\pi=(\pi_1, \dots, \pi_{\ell})$ of non-intersecting upright paths so that $\pi_r$ joins $(1,r)$ to $(n,m+r-\ell)$. In particular, 
\begin{equation}
\Zcal_{\ell}(n,m) = \det\left(\ZZ_1[(1,i)\vert(n,m+j-\ell)]\right)_{i,j=1}^{\ell}.
\end{equation} 
Following \cite{corwin2014tropical}, we also define variables $z_{m,\ell}(n)$ such that 
\begin{equation}
    \Zcal_{\ell}(n,m) =  z_{m,1}(n) \dots z_{m,\ell}(n).
    \label{eq:defsmallZ}
\end{equation}  
In other terms, $z_{m,1}(n) = \Zcal_1(n,m)$ and for $\ell\ge 2$, $z_{m,\ell}(n) = \Zcal_{\ell}(n,m)/\Zcal_{\ell-1}(n,m)$. 

\begin{proposition}[{\cite[Proposition 2.5]{corwin2014tropical}}]
Fix some integer $M$ (it will correspond to the maximal vertical coordinate, i.e. $1\le m\le M$). The array of numbers $(z_{m,\ell}(n))_{1\le \ell\le m\le M}$ evolves as follows as $n$ is increased. At each step, the recurrence involves an auxiliary array of numbers $a_{m,\ell}$. Given $(z_{m,\ell}(n-1))_{1\le \ell\le m\le M}$, one constructs $z_{m,\ell}(n)$ for increasing values of $\ell$ using the recurrence
\begin{subequations}
\begin{align}
	a_{m,1}&=w_{n,m}  &\textrm{for }& 1\le m\leq M\\
	a_{m+1,\ell+1}&=a_{m+1,\ell} \frac{z_{m+1,\ell}(n-1) z_{m,\ell}(n)}{z_{m+1,\ell}(n) z_{m,\ell}(n-1)}
	&\textrm{for }& 1\le \ell\leq m<M\\
	z_{m,\ell}(n)&= a_{m,\ell}\cdot \left(z_{m,\ell}(n-1)+z_{m-1,\ell}(n)\right)&\textrm{for }& 1\le \ell<m\leq M \label{eq:recurrencespecial}\\
	z_{m,m}(n)&= a_{m,m}\cdot z_{m,m}(n-1) &\textrm{for }& 1\le m\leq M.
\end{align}
\label{eq:recurrenceRSK}
\end{subequations}
\end{proposition}
The recursion \eqref{eq:recurrenceRSK} is one of the definitions of the so-called geometric Robinson-Schensted-Knuth (RSK) correspondence. The usual RSK correspondence is a bijection between matrices of non-negative integers and couples of semi-standard Young tableaux of the same shape, related to the representation theory of the symmetric group, whose definition is not important for our purposes, although we point out that it is known to be related to zero temperature models of directed polymers \cite{johansson2000shape}. A variant of the RSK correspondence, which can be thought of as a positive temperature analogue, was introduced by Kirillov \cite{Kri01} and is now referred to as the geometric RSK correspondence (see also \cite{NY04, corwin2014tropical}). This is a bijection between matrices of positive reals $(w_{i,j})_{1\leq i\leq n,1\leq j\leq M}$ and couples of arrays of positive real numbers $z=(z_{m,\ell}(n))_{1\le \ell\le m\le M}$  and $\tilde z=(\tilde z_{m,\ell}(M))_{1\le \ell\le m\le n}$. Under this bijection, the array $z$ is precisely defined by the recurrence relation \eqref{eq:recurrenceRSK} while the array $\tilde z$ can be defined by similar recurrence relations, essentially exchanging the roles of rows and columns, and we will not need it. 

The crucial fact that we will use below in order to find stationary non-intersecting polymer partition function is that the partition function $z_{m,\ell}(n)$ satisfies a recurrence relation \eqref{eq:recurrencespecial}, which takes a very similar form for any $\ell$ (recall that $\ell$ correspond to the number of non-intersecting paths). For $\ell=1$, this relation is trivial, it is equivalent to 
\begin{equation}
    \Zcal_1(n,m)=w_{n,m}(\Zcal_1(n-1,m)+\Zcal_1(n,m-1)).
    \label{eq:recurrencesinglepolymer}
\end{equation}
However, for $\ell>1$, the relation \eqref{eq:recurrencespecial} is not obvious at all. 

\subsection{Invariant measure for geometric RSK dynamics}

In this section, we will assume that the weights $w_{i,j}$ are distributed as $w_{i,j}\sim\Gammainv(\gamma_{i,j})$ and 
$\gamma_{i,j} = \alpha_i + \beta_j >0.$ For any $\gamma>0$, we denote by $\Gammainv(\gamma)$ the inverse gamma distribution with shape parameter $\gamma$ (and scale parameter $1$). The corresponding directed polymer model was first introduced in \cite{seppalainen2012scaling}, in the case where $\gamma_{i,j}$ are all identical, and in \cite{corwin2014tropical} in the general case. Let us define ratios of partition functions $u$ (horizontal ratios) and $v$ (vertical ratios) by 
\begin{equation} 
v_{m,\ell}(n) =  \frac{z_{m,\ell}(n)}{z_{m-1,\ell}(n)}\text{ and }u_{m,\ell}(n) =  \frac{z_{m,\ell}(n)}{z_{m,\ell}(n-1)}.
\label{eq:uvincrements}
\end{equation}
 For $\ell=1$ these ratios were denoted $U,V$ in \cite{seppalainen2012scaling}. The reason why this model was introduced in \cite{seppalainen2012scaling} is that one can find a stationary version of \eqref{eq:recurrencesinglepolymer} in the sense that if one imposes some special distribution on $\Zcal_1(n,1)$ and $\Zcal_1(1,m)$, i.e. on the boundaries of the quadrant $\mathbb Z_{>0}^2$, then all the $u$ and $v$ increment ratios will have the same distribution for any $n,m$ in the quadrant. This is based on the following property of Gamma random variables. 
 \begin{lemma}[{\cite[Lemma 3.2]{seppalainen2012scaling}}]
Let $U,V,w$ be independent random variables. Let 
\begin{equation}
U'=w\left(1+\frac{U}{V} \right), \;\; V'=w\left(1+\frac{V}{U} \right), \;\; w'=\left( \frac{1}{U} +\frac{1}{V}\right)^{-1}. 
\end{equation}
If for some  $\alpha>0$ and $\theta\in (-\alpha, \alpha)$,  $U\sim \mathrm{Gamma^{-1}}(\alpha+\theta)$, $V\sim \mathrm{Gamma^{-1}}(\alpha-\theta)$, $w\sim \mathrm{Gamma^{-1}}(2\alpha)$, then the triples $(U,V,w)$ and $(U', V', w')$ have the same distribution. 
\label{lem:seppalainenstationary}
\end{lemma}
For $\ell=1$, the recurrence relation \eqref{eq:recurrenceRSK} along with the definition of increment ratios \eqref{eq:uvincrements} imply that 
\begin{equation}
u_{m,1}(n) = w_{n,m}\left( 1 +\frac{u_{m-1,1}(n)}{v_{m,1}(n-1)} \right),\;\;\;v_{m,1}(n) =  w_{n,m}\left( 1 +\frac{v_{m,1}(n-1)}{u_{m-1,1}(n)} \right).
\label{eq:recurrenceforincrements}
\end{equation}
 These relation have exactly the same form as in Lemma \ref{lem:seppalainenstationary}. This imply that if for $m=1$ and for $n=1$, the ratios $u,v$ are independent and distributed as $u\sim \mathrm{Gamma^{-1}}(\alpha+\theta)$, $v\sim \mathrm{Gamma^{-1}}(\alpha-\theta)$, then this will be the case for all $n,m$.

 More generally, using \eqref{eq:recurrenceRSK}, for any $\ell$, the ratios $u,v$ satisfy 
 \begin{equation}
u_{m,\ell}(n) = a_{m,\ell}\left( 1 +\frac{u_{m-1,\ell}(n)}{v_{m,\ell}(n-1)} \right),\;\;\;v_{m,\ell}(n) =  a_{m,\ell}\left( 1 +\frac{v_{m,\ell}(n-1)}{u_{m-1,\ell}(n)} \right),
\label{eq:recurrenceforincrementsgeneral}
\end{equation}
where the $a_{m,\ell}$ are defined in \eqref{eq:recurrenceRSK}. If one can impose that these coefficients $a_{m,\ell}$ are independent and inverse gamma distributed, then we will obtain stationary distributions of increment ratios $u_{m,\ell}(n) ,v_{m,\ell}(n)$ for any $\ell$. This is indeed possible, and it was noticed in \cite{corwin2014tropical}. The following proposition summarizes the consequences. 
\begin{proposition}[{\cite[Theorem 3.10]{corwin2014tropical}}]
Assume that the array $\mathbf z(n)$ evolves according to the recurrence \eqref{eq:recurrenceRSK}. Let $k\geq 1$ and assume $\beta_1<\beta_2<\dotsm<\beta_k<\min\{\beta_{k+1}, \beta_{k+2, }\dotsc\}$. 

 Then the process
$\mathbf v(n) =  (v_{m,\ell})_{1\le \ell\le k,\, m>\ell}$ has an invariant distribution where the variables $v_{m,\ell}$ 
are independent with marginal distributions
$v_{m,\ell}\sim\Gammainv(\beta_m-\beta_\ell)$.
Moreover, if the process $\mathbf v(n) =  (v_{m,\ell})_{1\le \ell\le k,\, m>\ell}$ is started  with this distribution when $n=1$, then the variables
$$\{v_{m,\ell}(N) : 1\le\ell\le j,\, \ell<m\le M\}\cup
\{ u_{M,\ell}(n): 1\le n\le N, \, 1\le\ell\le j\}$$ are independent with
marginals $v_{m,\ell}(n)\sim\Gammainv(\beta_m-\beta_\ell)$ and
$ u_{M,\ell}(n)\sim\Gammainv(\alpha_n+\beta_\ell)$.
This also implies that along any down-right path in the $\mathbb Z_{>0}^2$ lattice, the $u,v$ increment ratios  are all independent, for $1\le \ell\le k$ and $(n,m)$ along the downright path, with marginals $v_{m,\ell}(n)\sim\Gammainv(\beta_m-\beta_\ell)$ and
$ u_{m,\ell}(n)\sim\Gammainv(\alpha_n+\beta_\ell)$. 
\label{prop:COSZstationary}
\end{proposition}
Since the evolution of the process $\mathbf z(n)$ was introduced in \eqref{eq:recurrenceRSK} to describe the evolution of partition functions of non-intersecting polymers (as the horizontal coordinate $n$ of arrival points increases), it is natural to ask whether the invariant measure in Proposition \ref{prop:COSZstationary} can be seen as ratios of partition functions of some polymer model with a certain distribution of weights. It is not clear to us if this can be done in general, but we will see in the next section that it will be possible in a certain limiting sense.

\subsection{Stationary non-intersecting log-gamma polymers} 
Using Proposition \ref{prop:COSZstationary}, we will define a stationary log-gamma polymer model for each integer $k\ge 1$. When $k=1$, the model will be exactly Sepp\"al\"ainen's stationary model from \cite{seppalainen2012scaling}. For $k>1$, the model will be such that the distribution of $\Zcal_{\ell}$ for $1\le\ell\le k$ is stationary (in the sense that the $u$ and $v$ increment ratios have the same distribution anywhere in the lattice).  
Recall the definition of the partition function  \eqref{eq:defpartitionfunctiongeneral}. 

\begin{definition}
\label{def:stationarymodel}
Fix $k\ge 1$ and $0<\lambda<\theta$. Let us define independent weights $w_{i,j}$ such that 
\begin{equation}
   w_{i,j} \sim \begin{cases}
    1 &\mbox{ if }1\le i,j\le k,\\
    \mathrm{Gamma}^{-1}(\lambda)&\mbox{ if }i> k, j\le k,\\
    \mathrm{Gamma}^{-1}(\theta-\lambda)&\mbox{ if }i\le k, j> k,\\
    \mathrm{Gamma}^{-1}(\theta)&\mbox{ if }i,j> k.
    \end{cases}
    \label{eq:stationaryweights}
\end{equation}
We define partition functions $\Zcal^{\rm stat}_{\ell}(n,m)$ for $1\le \ell\le k$ as follows. 
\begin{equation}  
    \Zcal^{\rm stat}_{\ell}(n,m) =  \ZZ_{\ell}\left[ (1,1), \dots, (1,\ell) \Big\vert (n,m), (n,m-1), \dots, (n,m-\ell+1)  \right],
    \label{eq:defZcalstat}
\end{equation} 
where we impose the extra condition that the $\ell$ non-intersecting polymers go through coordinates $(k-\ell+1,k), (k-\ell+2,k-1), \dots, (k,k-\ell+1)$, and all vertices in the triangle formed by the points $((k-\ell+1,k))$, $(k,k-\ell+1)$ and $(k,k)$ must belong to a path (see Fig. \ref{fig:Zstat}). We will call \textit{bottom-left packed}  such path configurations. 
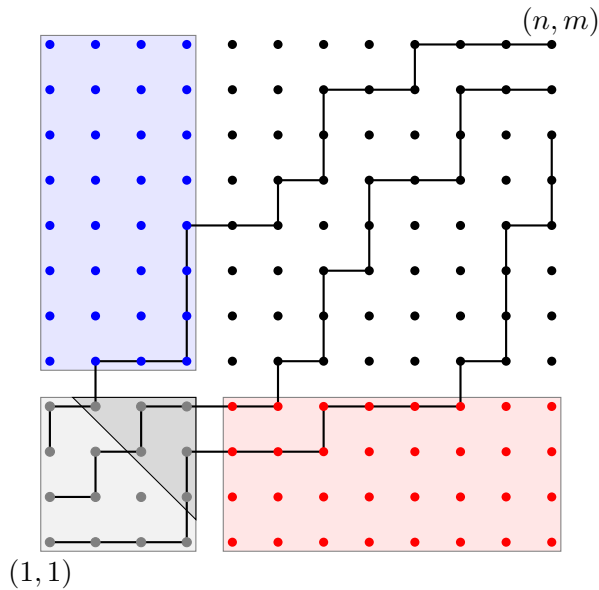
\begin{figure}
    \centering
    \begin{tikzpicture}[scale=0.6]
    \filldraw[fill=black!5!, draw=gray] (.8,.8) rectangle (4.2,4.2);
    \filldraw[fill=black!15!, draw=black] (1.5,4.2) -- (4.2,4.2) -- (4.2,1.5) -- cycle;
    \filldraw[fill=blue!10!, draw=gray] (.8,4.8) rectangle (4.2,12.2);
    \filldraw[fill=red!10!, draw=gray] (4.8,.8) rectangle (12.2,4.2);

    \draw[thick] (1,3) -- (1,4) -- (2,4) -- (2,5) -- (4,5) -- (4,8) -- (6,8) -- (6,9) -- (7,9) -- (7,11) -- (9,11) -- (9,12) -- (12,12);
    \draw[thick] (1,2) -- (2,2) -- (2,3) -- (3,3) -- (3,4) -- (6,4) -- (6,5) -- (7,5) -- (7,7)-- (8,7) -- (8,9) -- (9,9) -- (10,9) -- (10,11) --  (12,11);
    \draw[thick] (1,1)-- (4,1) -- (4,2) -- (4,3) -- (7,3) -- (7,4) -- (10,4) -- (10,5) -- (11,5) -- (11,8) -- (12,8)-- (12,10);
    \foreach \x in {1,...,4}
    {\foreach \y in {1,...,4}
    \filldraw[gray](\x,\y) circle(0.1);
    }
    \foreach \x in {1,...,4}
    {\foreach \y in {5,...,12}
    \fill[blue](\x,\y) circle(0.1);
    }
    \foreach \x in {5,...,12}
    {\foreach \y in {1,...,4}
    \fill[red](\x,\y) circle(0.1);
    }
    \foreach \x in {5,...,12}
    {\foreach \y in {5,...,12}
    \fill(\x,\y) circle(0.1);
    }
    
    \draw (0.8,0.3) node{$(1,1)$};
    \draw (12.2,12.5) node{$(n,m)$};
\end{tikzpicture}
    \caption{The partition function $\Zcal^{\rm stat}_{\ell}(n,m)$ is the sum of paths as shown in the figure. Here we have set $k=4$, $\ell=3$, $n=m=12$. Points in the gray area have weight $1$. Points in the blue (resp. red) area have weights distributed as $\Gammainv(\theta-\lambda)$ (resp. $\Gammainv(\lambda)$). All other weights have weights distributed as $\Gammainv(\theta)$. In the triangular region with a darker gray shade, all vertices must belong to a path.}
    \label{fig:Zstat}
\end{figure}
This condition that paths are bottom-left packed is equivalent to imposing that paths go through a maximal amount of points inside the $k\times k$ bottom left corner. We will see why this condition is needed below.

As in \eqref{eq:defsmallZ}, we define
$z^{\rm stat}_{m,1}(n) = \Zcal^{\rm stat}_1(n,m)$ and for $\ell\ge 2$, $z^{\rm stat}_{m,\ell}(n) = \Zcal^{\rm stat}_{\ell}(n,m)/\Zcal^{\rm stat}_{\ell-1}(n,m)$, so that 
\begin{equation}
    \Zcal^{\rm stat}_{\ell}(n,m) =  z^{\rm stat}_{m,1}(n) \dots z^{\rm stat}_{m,\ell}(n).
    \label{eq:defsmallZstat}
\end{equation}  
We also define horizontal/vertical increments ratios as usual by 
\begin{equation}
    u_{\ell}^{\rm stat}(n,m)=\frac{z^{\rm stat}_{\ell}(n,m)}{z^{\rm stat}_{\ell}(n-1,m)}, \;\;\; v_{\ell}^{\rm stat}(n,m)=\frac{z^{\rm stat}_{\ell}(n,m)}{z^{\rm stat}_{\ell}(n,m-1)}
    \label{eq:defuvstat}
\end{equation}
\end{definition}

\begin{proposition}
\label{prop:stationarymodel}
Fix any down-right path in the quadrant $\mathbb Z_{>0}^2$. The increments ratios $u^{\rm stat}_{\ell}(n,m)$ and $v^{\rm stat}_{\ell}(n,m)$ along that downright path are all independent for $1\le \ell\le k$, and distributed as 
$$ u^{\rm stat}_{\ell}\sim \mathrm{Gamma}^{-1}(\lambda), \;\;  v^{\rm stat}_{\ell}\sim \mathrm{Gamma}^{-1}(\theta- \lambda).$$
\end{proposition}
\begin{proof}
Let us consider the log-gamma polymer model as in the previous section where $w_{i,j}\sim\mathrm{Gamma}^{-1}(\alpha_i+\beta_j)$. Let us  set $\alpha_i=\beta_i=i\eps$ for $1\le i\le k$, $\alpha_i=\lambda$ for $i>k$, and $\beta_j=\theta-\lambda$ for $j>k$ and call $\Zcal_{\ell}^{\rm \eps-stat}$ the partition functions as defined in \eqref{eq:partitionfunctionell}. On the one hand, we claim that as $\eps$ goes to zero, for $n,m>k$, we have the convergence in distribution 
\begin{equation}
     \eps^{\ell(2k-\ell)} \prod_{i=1}^\ell \frac{(2k-i+1)!}{i!} \Zcal_{\ell}^{\rm \eps-stat}(n,m) \xRightarrow[\eps \to 0]{} \Zcal_{\ell}^{\rm stat}(n,m),
     \label{eq:limitZepsstat}
\end{equation}
where $\Zcal_{\ell}^{\rm stat}(n,m)$ is defined in \eqref{eq:defZcalstat}. Indeed, it is clear that outside of the $k\times k$ bottom-left corner, the distribution of weights in $\Zcal_{\ell}^{\rm \eps-stat}(n,m)$ converges to the distribution of weights in $\Zcal_{\ell}^{\rm stat}(n,m)$. In the the $k\times k$ bottom-left corner, the weights in $\Zcal_{\ell}^{\rm \eps-stat}(n,m)$ are however diverging, but one has that $\epsilon (i+j) w_{i,j}$ converges to $1$. As $\epsilon$ goes to zero, the sum over all non-intersecting  paths in $\Zcal_{\ell}^{\rm \eps-stat}(n,m)$ can be approximated by the sum over non-intersecting paths that collect a maximal number of diverging weights $w_{i,j}$ for $1\leq i,j\leq k$. This is why we had imposed the condition that paths are bottom-left packed in the definition of $\Zcal_{\ell}^{\rm stat}(n,m)$ above. Moreover, the prefactor in \eqref{eq:limitZepsstat} is the correct factor to normalize the diverging weights (it does not depend on the choice of paths, as long as paths collect a maximal number diverging weights in the $k\times k$ bottom-left corner).

On the other hand, we may use Proposition \ref{prop:COSZstationary} in the quadrant $\mathbb Z_{\geq k}^2$. Along the $k$th row, if we restrict the partition function $Z_{\ell}^{\rm \eps-stat}$ to a sum over paths that collect a maximal number of diverging weights, the distribution of $u$ ratios is such that   $u_{k,\ell}(n)\sim\Gammainv(\lambda+(k-\ell+1)\epsilon)$ (for $n>k$). Similarly, we find that the distribution of $v$ ratios along the $k$-th column is such that  $v_{m,\ell}(k)\sim\Gammainv(\theta-\lambda+(k-\ell+1)\epsilon)$ (for $m>k$). This distribution of $v$ ratios is not exactly the one that is stationary according to Proposition \ref{prop:COSZstationary}, however, they become identical as $\eps$ go to zero. This implies that along any down right path in the quadrant $\mathbb Z_{\geq k}^2$, and up to a $O(\eps)$ error, the distribution of $u,v$ ratios for the partition function $\Zcal_{\ell}^{\rm \eps-stat}$ is the same as in the statement of Proposition \ref{prop:stationarymodel}. Finally, \eqref{eq:limitZepsstat} implies that the $u,v$ ratios of $\Zcal_{\ell}^{\rm \eps-stat}$ and $\Zcal_{\ell}^{\rm stat}$ have the same distribution as $\eps$ goes to zero, which implies the statement of the Proposition. 
\end{proof}

Now, fix some $k\geq 1$. For a family of integers $k\leq p_1<\dots<p_\ell$, we will define the partition function (slightly overloading notation) 
\begin{equation}
 \Zcal^{stat}_{\ell}(\mathbf p,m)  = \ZZ_{\ell}((1,1), \dots, (\ell,1) \vert  (p_1,m), \dots, (p_{\ell},m)) 
 \label{eq:defZstatp}
\end{equation}
where the weights have been chosen as in the stationary partition function, i.e.  as in \eqref{eq:stationaryweights}, and the paths must satisfy the same conditions as in Definition \ref{eq:defZcalstat} (see Fig. \ref{fig:Zstatmultipoint}).  \begin{figure}
    \centering
    \begin{tikzpicture}[scale=0.6]
    \filldraw[fill=black!5!, draw=gray] (.8,.8) rectangle (4.2,4.2);
    \filldraw[fill=black!15!, draw=black] (1.5,4.2) -- (4.2,4.2) -- (4.2,1.5) -- cycle;
    \filldraw[fill=blue!10!, draw=gray] (.8,4.8) rectangle (4.2,7.2);
    \filldraw[fill=red!10!, draw=gray] (4.8,.8) rectangle (12.2,4.2);
    \draw[thick] (1,1) -- (1,4) -- (2,4) -- (2,5) -- (4,5) -- (4,7);
    \draw[thick] (2,1) -- (2,3) --  (3,3) -- (3,4) -- (6,4) -- (6,5) -- (7,5) -- (7,7)-- (8,7);
    \draw[thick] (3,1) -- (4,1) -- (4,2) -- (4,3) -- (7,3) -- (7,4) -- (10,4) -- (10,5) -- (11,5) -- (11,7) ;
    \foreach \x in {1,...,4}
    {\foreach \y in {1,...,4}
    \filldraw[gray](\x,\y) circle(0.1);
    }
    \foreach \x in {1,...,4}
    {\foreach \y in {5,...,7}
    \fill[blue](\x,\y) circle(0.1);
    }
    \foreach \x in {5,...,12}
    {\foreach \y in {1,...,4}
    \fill[red](\x,\y) circle(0.1);
    }
    \foreach \x in {5,...,12}
    {\foreach \y in {5,...,7}
    \fill(\x,\y) circle(0.1);
    }
    \draw (0.8,0.3) node{$(1,1)$};
        \draw (4,7.6) node{$(p_1,m)$};
    \draw (8,7.6) node{$(p_2,m)$};
    \draw (11,7.6) node{$(p_3,m)$};

\end{tikzpicture}
    \caption{The partition function $\Zcal^{\rm stat}_{\ell}(\mathbf p,m)$ is the sum of paths as shown in the figure. Here we have set $\mathbf p=(4,8,11)$, $k=4$, $\ell=3$, $m=7$. Points in various colors have the same weights as in Figure \ref{fig:Zstat}. In the triangular region with a darker gray shade, all vertices must belong to one path. Here, by convention, path start from $(1,1), \dots, (\ell,1)$, although they could start anywhere in the gray region provided the gray triangle is full, since the weights there are equal to $1$. }
    \label{fig:Zstatmultipoint}
\end{figure}
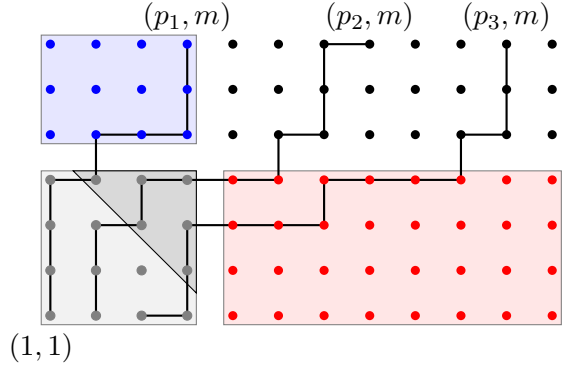
Note that the starting points of polymers are now $(1,1), \dots, (\ell,1) $ instead of $(1,1), \dots, (1, \ell)$ as in \eqref{eq:defZcalstat}. This has no influence on the partition function, since all weights of the form $(i,j)$ with $i+j\leq \ell+1$ must belong to a path in both cases, due to the non-intersection condition. 

\begin{proposition}
Fix $k\geq 1$. 
The family  of partition function ratios 
\begin{equation}  \left\lbrace  \frac{\Zcal^{stat}_{\ell}(\mathbf p,m)}{ \Zcal^{stat}_{\ell}(\boldsymbol{\delta},m)} \right\rbrace_{\footnotesize \begin{matrix} 1\leq \ell\leq k\\ p_1<\dots<p_\ell \end{matrix}},
\label{eq:ratiomultipoint}
\end{equation}
where $\boldsymbol{\delta} =(1,2,\dots, k)$, has a joint distribution which does not depend on $m$ for $m\geq k$. In particular, 
\begin{equation}
    \frac{\Zcal^{stat}_{\ell}(\mathbf p,m)}{ \Zcal^{stat}_{\ell}(\boldsymbol{\delta},m)} \overset{(d)}{=} \mathbf Z[(k,k), \dots, (k,k-\ell+1) \vert (p_1,k),\dots, (p_\ell,k)],
    \label{eq:stationarymultipoint}
\end{equation}
where the partition function $\mathbf Z$ is defined with respect to weights which are i.i.d. and $\mathrm{Gamma}^{-1}(\lambda)$ distributed (see Fig. \ref{fig:Zstatmultipoint2}). 
\label{prop:log-gammaarbitrarymultipoint}
\end{proposition}
\begin{figure}
    \centering
    \begin{tikzpicture}[scale=0.6]
    \filldraw[fill=black!5!, draw=gray] (.8,.8) rectangle (4.2,4.2);
    \filldraw[fill=black!15!, draw=black] (1.5,4.2) -- (4.2,4.2) -- (4.2,1.5) -- cycle;
    \filldraw[fill=red!10!, draw=gray] (4.8,.8) rectangle (18.2,4.2);
    \draw[thick] (1,1) -- (1,4) -- (2,4) -- (7,4);
    \draw[thick]   (2,1) -- (2,3) -- (3,3) -- (9,3) -- (9,4) --  (12,4);
    \draw[thick] (3,1) -- (4,1) -- (4,2) -- (11,2) -- (11,3) -- (14,3) -- (14,4) -- (16,4) ;
    \foreach \x in {1,...,4}
    {\foreach \y in {1,...,4}
    \filldraw[gray](\x,\y) circle(0.1);
    }
    \foreach \x in {5,...,17}
    {\foreach \y in {1,...,4}
    \fill[red](\x,\y) circle(0.1);
    }
    \draw (0.8,0.3) node{$(1,1)$};
        \draw (7,4.6) node{$(p_1,k)$};
    \draw (12,4.6) node{$(p_2,k)$};
    \draw (16,4.6) node{$(p_3,k)$};

\end{tikzpicture}
    \caption{The ratio of partition functions \eqref{eq:ratiomultipoint} have the same distribution as the partition function shown in the picture, which corresponds to the case $m=k$.   Here we have set $\mathbf p=(7,12,16)$, $k=4$, $\ell=3$. Since the weights in the gray region have weight $1$, the paths may start from the vertices $(k,k), \dots, (k,1)$ as in the definition of the partition function $\mathbf Z[(k,k), \dots, (k,1) \vert (p_1,k),\dots, (p_k,k)]$ in \eqref{eq:stationarymultipoint}.}
    \label{fig:Zstatmultipoint2}
\end{figure}
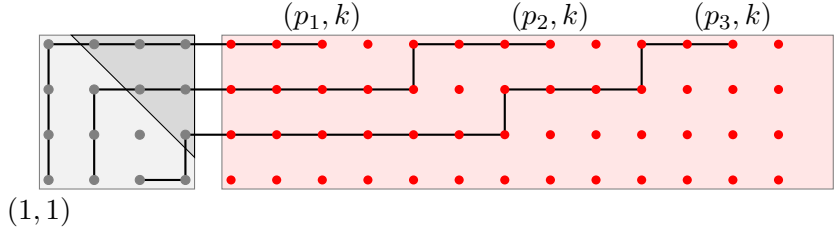
\begin{proof}
Fixed $m\geq k$ and consider the model $\Zcal^{\eps-stat}$, as in the proof of Proposition \ref{prop:stationarymodel}.  For any $\mathbf p$, the partition function $\Zcal^{\eps-stat}_{\ell}(\mathbf p,m)$ can be written as a deterministic function of the family of partition functions $\lbrace \Zcal^{\rm stat}_\ell(n,m)\rbrace_{n\geq \ell}$. The specific relation is a bit involved, and it is  given explicitly as Theorem 1.1 in \cite{corwin2020invariance}. What matters to us here is that this property implies in turn that the ratios \eqref{eq:ratiomultipoint} can be written as functions of ratios of the form $\Zcal^{\rm stat}_\ell(n_1,m)/\Zcal^{\rm stat}_\ell(n_2,m)$. By Proposition \ref{prop:stationarymodel}, the joint distribution of the latter partition function ratios does not depend on $m$. Thus, the joint distribution of the ratios \eqref{eq:ratiomultipoint} does not depend on $m$ neither. 
\end{proof}
\begin{remark}
In \cite{barraquand2021steady}, stationary measures for the (single polymer) partition function of the log-gamma polymer in a half-quadrant was obtained using the symmetry of its distribution with respect to permuting inhomogeneity parameters. A similar argument can be carried out in the present context. Indeed, in the partition function $\Zcal_{\ell}^{\rm \eps-stat}(n,m)$ introduced in the proof of Proposition \ref{prop:stationarymodel}, one may exchange the order of parameters $\alpha_i, \beta_j$ without changing the distribution of $\Zcal_{\ell}^{\rm \eps-stat}(n,m)$ (this is a consequence of more general results from \cite{barraquand2020halfMacdonald}). In particular, exchanging the parameters $\lbrace \beta_i\rbrace_{1\leq i\leq k}$ with parameters  $\lbrace \beta_{m-k+i}\rbrace_{1\leq i\leq k}$ would yield another proof of Proposition \ref{prop:log-gammaarbitrarymultipoint} without appealing to Proposition \ref{prop:COSZstationary} (which comes from properties of the RSK algorithm studied in \cite{corwin2014tropical}).
\end{remark}
\subsection{Interpretation in terms of very long polymers} 
\label{sec:sendingwithanangle}

The distribution of ratios of partition functions observed in stationary models are the same as the ratios that arise in models with homogeneous parameters for very long polymers (for the case of single polymers, this was rigorously proved in \cite{georgiou2013ratios}).

Let us start with the case $\ell=1$. We know that if for all $i,j>1$, $w_{i,j}\sim\Gammainv(\theta)$ and $w_{1,j}\sim\Gammainv(\theta-\lambda)$, $w_{i,1}\sim\Gammainv(\lambda)$, the ratios of partition functions are given by $u\sim \Gammainv(\lambda), v\sim\Gammainv(\theta-\lambda)$ \cite{seppalainen2012scaling}. The same distribution of ratios is obtained in a model with homogeneous weigths $w_{i,j}\sim\Gammainv(\theta)$ (where $i,j\in \mathbb Z$) if the starting point is sent to $\infty$ in the direction \cite[Theorem 4.1]{georgiou2013ratios}
 \begin{equation}
-(\Psi_1(\theta-\lambda), \Psi_1(\lambda)).
\label{eq:direction}
 \end{equation}  
More concretely, let us consider a polymer model as in Section \ref{sec:notations} with weights $w_{i,j}\sim \mathrm{Gamma}^{-1}(\theta)$ for any $i,j\in \mathbb Z$ and recall that we denote its partition function by $\mathbf Z_1[\mathbf x\vert \mathbf y]$. Let us also consider the stationary model from Definition \ref{def:stationarymodel} with $k=1$, whose partition function is denoted by $\Zcal_{1}^{\rm stat}(n,m)$. The claim above states that as $\mathbf x=-N(\Psi_1(\rho-\lambda), \Psi_1(\lambda)) $ and $N$ goes to infinity, the distribution of increment ratios 
\begin{equation}
    u_{m,1}(n)=\frac{\mathbf Z_1[\mathbf x \vert (n,m)]}{\mathbf Z_1[\mathbf x\vert (n-1,m)]}\;\;  \text{ and }  \;\; v_{m,1}(n) = \frac{\mathbf Z_1[\mathbf x\vert (n,m)]}{\mathbf Z_1[\mathbf x\vert (n,m-1)]}
\end{equation}
become the same as increment ratios of the partition functions $\Zcal_{1}^{\rm stat}(n,m)$, that is 
\begin{equation}
     u_{m,1}(n) \sim\Gammainv(\lambda),\;\;\;\; v_{m,1}(n)\sim\Gammainv(\theta-\lambda),
\end{equation}
(and this remains true jointly for $n,m$ in some bounded domain). 
\begin{remark}
This can be also interpreted using the fact that in the partition function $\mathcal Z_{1}^{\rm stat}(n,m)$, if one sends parameter $(n,m)$ to $\infty$ in the direction $(\Psi_1(\rho-\lambda), \Psi_1(\lambda))$, a simple computation, using the known explicit value of the free energy, shows that the optimal path will detach from the first row at a location such that the detaching point and the arrival point form an angle as defined by the vector \eqref{eq:direction}. See similar arguments in \cite{baik2005phase, barraquand2020halfMacdonald, krajenbrink2021tilted}.
\end{remark}

More generally, for  several  non-intersecting polymers, starting at positions 
 \begin{equation}
\mathbf x_{\ell}=(-N\Psi_1(\theta-\lambda_{\ell})+1, -N\Psi_1(\lambda_{\ell}) +\ell), \;\;\;1\leq \ell\leq k,
\label{eq:directiongeneral}
\end{equation} 
we may define ratios $u_{m\ell}(n), v_{m,\ell}(n)$ (similarly as before) by 
\begin{align}
    u_{m,1}(n) \dots u_{m\ell}(n) &= \frac{\mathbf Z [\mathbf x_1, \dots, \mathbf x_{\ell} \vert (n,m-\ell+1), \dots, (n,m)]}{\mathbf Z [\mathbf x_1, \dots, \mathbf x_{\ell} \vert (n-1,m-\ell+1), \dots, (n-1,m)]}\\
    v_{m,1}(n) \dots v_{m\ell}(n) &=\frac{\mathbf Z [\mathbf x_1, \dots, \mathbf x_{\ell} \vert (n,m-\ell+1), \dots, (n,m)]}{\mathbf Z [\mathbf x_1, \dots, \mathbf x_{\ell} \vert (n,m-\ell), \dots, (n-1,m-1)]}.
\end{align}
Then as $N$ goes to infinity, we expect that 
the increment ratios will be distributed as  
\begin{equation} u_{m,\ell}(n) \sim\Gammainv(\lambda),\;\;\;\; v_{m,\ell}(n)\sim\Gammainv(\theta-\lambda),\label{eq:stationarydirectiongeneral}
\end{equation} 
and all independent for $1\leq \ell \leq k$ and $(n,m)$ varying along a downright path in the lattice.

\section{KPZ limit}
\label{sec:KPZlimit}
In this Section, we will take a discrete to continuous limit to deduce our main results, stated in Section \ref{sec:mainresults}, from the results about the log-gamma polymer obtained  in Section \ref{sec:discretecase}. 

\subsection{Convergence of discrete (non-intersecting) directed polymers to continuous limits}
\label{sec:discretetoKPZ}

Discrete directed polymers such as discussed in Section \ref{sec:discretecase} converge to the KPZ equation as the temperature parameter $\theta$ goes to $\infty$ simultaneously with the length of the polymer paths. This was first proved in \cite{alberts2012intermediate} for directed polymer models with general weight distribution, under mild assumptions. The specific case of the log-gamma directed polymer was discussed in \cite{corwin2017intermediate}, where it is also proved that partition function for non-intersecting discrete polymers converge to the O'Connell-Warren multilayer stochastic heat equation (see more details below). This convergence has been reused in many works to deduce results about the KPZ equation from results about the log-gamma polymer. Hence, we will simply state the needed convergence results,  adopting notations close to \cite{barraquand2020stochastic, barraquand2020half, barraquand2022stationary}. 

\bigskip

Let us start with the case $\ell=1$. Consider a discrete polymer model as in Section \ref{sec:notations} with i.i.d. weights distributed as $w_{i,j}\sim\mathrm{Gamma}^{-1}(\theta)$. As $\mathbb E[w_{i,j}]=\frac{1}{\theta -1}$, we will let $\theta=1+2\sqrt{n}$ and multiply all weights by $\sqrt{n}$ to normalize their expectation to $1/2$. Then, we need to scale appropriately the starting and ending points so that discrete up-right paths become Brownian paths. Hence, for $s<t$ and $x,y\in \mathbb R$, we scale the starting and ending points  as 
\begin{equation}\mathbf x^{(n)} = (1+sn/2+x\sqrt{n},1+sn/2), \quad\quad \mathbf y^{(n)}= (tn/2+y\sqrt{n},tn/2), \end{equation} 
and denote by $L(\mathbf x^{(n)}, \mathbf y^{(n)}) = 1+(t-s)n+(y-x)\sqrt{n} $ the number of vertices along a lattice path from $\mathbf x^{(n)}$ to $\mathbf y^{(n)}$. Then we have the convergence in distribution 
\begin{equation}
    \sqrt{n}^{L(\mathbf x^{(n)}, \mathbf y^{(n)})} \mathbf Z_1[\mathbf x^{(n)} \vert \mathbf y^{(n)} ] \xRightarrow[n\to\infty]{} Z(x,s\vert y,t)
    \label{eq:convsinglepolymers}
\end{equation}
where $Z(x,s\vert y,t)$ is defined in \eqref{eq:defpointopointcontinuousDP}. Such convergence result was initially obtained in \cite{alberts2012intermediate}. 

\bigskip  

Assume now that $\ell$ is arbitrary and consider $\vec x=x_1 < x_2<\dots <x_\ell$ and $\vec y=y_1<y_2<\dots < y_\ell$.  We scale starting and ending points as 
\begin{equation}\mathbf x_i^{(n)} = (1+sn/2+x_i\sqrt{n},1+sn/2), \quad\quad \mathbf y_i^{(n)}= (tn/2+y_i\sqrt{n},tn/2), \quad \quad 1\leq i\leq \ell.\end{equation} 
Then, we have the convergence in distribution 
\begin{equation}
    \sqrt{n}^{\sum_{i=1}^{\ell} L(\mathbf x_i^{(n)}, \mathbf y_i^{(n)})} \;\mathbf Z_{\ell}[\mathbf x_1^{(n)}, \dots,\mathbf x_{\ell}^{(n)} \vert \mathbf y_1^{(n)}, \dots,  \mathbf y_{\ell}^{(n)}] \xRightarrow[n\to\infty]{} Z(\vec x,s\vert \vec y,t)
\end{equation}
where $Z(\vec x,s\vert \vec y,t)$ was defined in \eqref{eq:defZl}. This is simply a consequence of the convergence \eqref{eq:convsinglepolymers} along with the Karlin-McGregor formulas \eqref{eq:defZKarlinMcGregor} and \eqref{eq:defpartitionfunctiongeneral} (for continuous and discrete polymers respectively).

\bigskip 
Consider now the partition function $\Zcal_{\ell}(n,m)$ defined in \eqref{eq:partitionfunctionell}, where the $\ell$ non-intersecting polymers start and end from neighbouring locations. It is proved in \cite[Theorem 1.8]{corwin2017intermediate} that these partition functions converge to the O'Connell-Warren multilayer stochastic heat equation. More precisely, for  $\mathbf x^{(n)} = (tn/2+\ell+x\sqrt{n}, tn/2+\ell)$, we have the convergence in distribution 
 \begin{equation}
    \frac{\sqrt{n}^{\ell L((1,\ell),\mathbf x^{(n)})}\Zcal_{\ell}(tn/2+\ell+x\sqrt{n}, tn/2+\ell) }{\left(\frac{tn+x\sqrt{n}}{2}\right)^{-\ell^2/2}\prod_{j=0}^\ell j!} \xRightarrow[n\to\infty]{} M_{\ell}(t, \vec 0, x\vec 1),
    \label{eq:convO'ConnellWarren}
 \end{equation}
 where $M_{\ell}(t, \vec 0, x\vec 1)$ was defined in Section \ref{sec:method} and also denoted $Z_{\ell}(0,0\vert x,t)$. Note that we have slightly extended the result of \cite[Theorem 1.8]{corwin2017intermediate} which correspond to the case $x=0, t=2$.

\subsection{Invariant measures of the O'Connell-Warren multilayer SHE}
\label{sec:invariantOConnellWarren}
Recall the definition of ratios $R_{\ell}(t,x,y)$ from \eqref{eq:defratiosRell}. We will also more generally denote by $\left( R_{\ell}(y,t)\right)_{1\le \ell\le k}$ a process having the same Markovian evolution, with some unspecified initial condition. Proposition \ref{prop:stationarymodel} and the convergence from \eqref{eq:convO'ConnellWarren} yield invariant measures for this Markovian evolution.

More precisely, ratios $R_{\ell}(t,x,y)/R_{\ell}(t,x,z)$ are given by scaling limits of discrete partition function ratios 
$$ \frac{\mathcal Z_{\ell}(n,m)}{\mathcal Z_{\ell-1}(n,m)},$$
under the scalings given in \eqref{eq:convO'ConnellWarren}. 
Consider now the stationary model $\mathcal Z_\ell^{\rm stat}(n,m)$ from Definition \ref{def:stationarymodel}. 
Using the definition of increment ratios  \eqref{eq:defuvstat}, we may write
$$ \frac{\mathcal Z_{\ell}^{\rm stat}(tn/2+\ell + x\sqrt{n},tn/2+\ell)}{\mathcal Z^{\rm stat}_{\ell-1}(tn/2+\ell + x\sqrt{n},tn/2+\ell)} =  \frac{\mathcal Z^{\rm stat}_{\ell}(tn/2+\ell ,tn/2+\ell)}{\mathcal Z^{\rm stat}_{\ell-1}(tn/2+\ell ,tn/2+\ell)} \prod_{k=1}^{x\sqrt{n}} \frac{u_{\ell}^{\rm stat}(tn/2+\ell+k,tn/2+\ell)}{u_{\ell-1}^{\rm stat}(tn/2+\ell+k,tn/2+\ell)}.$$
Using Proposition  \ref{prop:stationarymodel} we know that the $u_{\ell}^{\rm stat}$ variables appearing in the equation above are all independent and distributed as inverse Gamma random variables with parameter $\lambda$.  
It is easy to see that, under the same scalings as in \eqref{eq:convO'ConnellWarren}, a multiplicative random walks with inverse Gamma distributed increments with parameter $\lambda$ converges to exponentials of Brownian motions with the same drift (for that one needs to scale $\lambda$ as $\lambda=\frac 1 2 + \sqrt{n} - \nu$ for some constant $\nu\in \mathbb R$ which will be the drift of the Brownian motions).

Hence,  we obtain that if at time $t=0$, $\left( R_{\ell}(y,0)\right)_{1\le \ell\le k}$ is given by a collection of independent exponentials of  Brownian motions with the same drift, then, for any time $t>0$, the processes $\left( R_{\ell}(y,t)/R_{\ell}(0,t)\right)_{1\le \ell\le k}$ are also distributed as independent exponentials of Brownian motions with the same drift.

We expect that these invariant measures can be observed in the local fluctuations of the fields at large time. More precisely, for all $\ell\ge 1$, we expect that 
\begin{equation}
\lim_{t\to\infty} \frac{R_{\ell}(t, 0, y) }{R_{\ell}(t, 0, 0)}\overset{(d)}{=} e^{B_\ell(y)},
\label{eq:stationarymeasureforR}
\end{equation}
where the (two-sided) Brownian motions $B_\ell(y)$ are independent and have drift $0$. Moreover, we expect that the result remains true for any fixed starting points $\vec x$ (not necesarily equal). In this case, the result needs to be reformulated in terms of $M_{\ell}(t,\vec x, \vec y)$ and reads 
 \begin{equation}
\lim_{t\to\infty} \frac{\frac{M_{\ell}(t, \vec x, y\vec 1)}{M_{\ell-1}(t, \vec x, y\vec 1)} }{\frac{M_{\ell}(t, \vec x, \vec 0)}{M_{\ell-1}(t, \vec x, \vec 0)}}\overset{(d)}{=} e^{B_\ell(y)}. 
\label{eq:stationarymeasureforM}
\end{equation}

\subsection{Partition function at arbitrary arrival points (first derivation of the main result)} 

In this section we will establish the main result of this paper, that is equation \eqref{eq:multipolymers}. We will show that \eqref{eq:multipolymers} can be deduced from \eqref{eq:stationarymeasureforR} using results from \cite{o2016multi}. For simplicity, we start with the case $\ell=2$. Using \cite[Proposition 6.3]{o2016multi}, we have 
\begin{equation}
\frac{M_2(t,(x,x),(y_1,y_2))}{M_1(t,x,y_1)M_1(t,x,y_2)} = \frac{1}{y_2-y_1}\int_{y_1}^{y_2} \frac{M_2(t,(x,x),(z,z))}{M_1(t,x,z)^2}\mathrm d z.
\label{eq:arbitrarypointsfromsame}
\end{equation}
Under the stationary measure (or, equivalently, as $t$ goes to infinity, see \eqref{eq:stationarymeasureforM}), we expect that for any fixed $x\in \R$, 
$$\frac{M_2(t,(x,x),(z,z))}{M_1(t,x,z)} = e^{B_2(z)} \frac{M_2(t,(x,x),(0,0))}{M_1(t,x,0)} \;\;\;\text{ and }\;\;\; M_1(t,x,z)=  e^{B_1(z)}M_1(t,x,0),$$
where $B_1, B_2$ are independent Brownian motions. 
Thus, for fixed $x$, we have the following equality in distribution between processes in variables $y_1,y_2$ with $y_1\le y_2$,
\begin{equation}
\lim_{t\to+\infty}\frac{M_2(t,(x,x),(y_1,y_2))}{M_2(t,(x,x),(0,0))} = \frac{e^{B_1(y_1)+B_1(y_2)}}{y_2-y_1}\int_{y_1}^{y_2} e^{B_2(z)-B_1(z)}\mathrm d z.
\label{eq:densityarrivalell=2}
\end{equation}

We now turn to the general case.  
Using \cite[Theorem 3.4]{o2016multi}, we obtain that for $y_1\le \dots\le y_\ell$, fixed $x\in \R$, $t>0$, we have 
\begin{equation}
M_\ell(t, x\vec 1, \vec y) = \frac{\prod_{j=1}^{\ell-1}j!}{\Delta(\vec y)}\prod_{i=1}^\ell M_1(t,x,y_i) \int_{GT(\vec y)} \prod_{k=1}^{\ell-1} \prod_{i=1}^{\ell-k} \frac{M_{k+1}(t, x\vec 1, z_i^{\ell-k}\vec 1)/M_{k}(t, x\vec 1, z_i^{\ell-k}\vec 1)}{M_{k}(t, x\vec 1, z_i^{\ell-k}\vec 1)/M_{k-1}(t, x\vec 1, z_i^{\ell-k}\vec 1)}\mathrm d z_i^{\ell-k}
\label{eq:theorem3.4}
\end{equation} 
where  $GT(\vec y)$ is the set of vectors $z^{1}, z^2, \dots , z^{\ell-1}\in \R\times \R^2\times  \dots\times \R^{\ell-1}$ such that $z^{1}\prec \dots \prec z^{\ell-1}\prec \vec y$ (where $\prec$ means interlacing) and by convention $M_0=1$. In order to match our notations with \cite[Theorem 3.4]{o2016multi} we have used the same notation for $M_n$ (at least when coordinates are distinct),  and the determinant appearing in the right-hand side of \eqref{eq:relationMndeterminant} is denoted $\tau_n(t,x,y)$ in \cite{o2016multi}. It seems to us that there is a minor mistake about the constants appearing in \cite{o2016multi}: there seems to be a constant missing in \cite[Eq. (39)]{o2016multi} to ensure continuity of $M_n$, and this missing constant would also be missing in the left-hand side of \cite[Theorem 3.4]{o2016multi}. It was taken into account in the statement \eqref{eq:theorem3.4} above. 

Using that under the stationary measure  (or, equivalently, as starting points are sent to infinity, that is $t\to +\infty$), 
$$ R_{k}(t, x, z) = R_{k}(t, x, 0)e^{B_k(z)},$$
where the $B_k$ are independent Brownian motions, or equivalently writing 
$$ \frac{M_{k}(t, x\vec 1, z\vec 1)}{M_{k-1}(t, x\vec 1, z\vec 1)} = \frac{M_{k}(t, x\vec 1, \vec 0)}{M_{k-1}(t, x\vec 1, \vec 0)}e^{B_k(z)},$$
we arrive at 
\begin{equation}
\lim_{t\to+\infty} \frac{M_\ell(t, x\vec 1, \vec y)}{M_\ell(t, x\vec 1, \vec 0)} = \frac{\prod_{j=1}^{\ell-1}j!}{\Delta(\vec y)}\prod_{i=1}^\ell e^{B_1(y_i)} \int_{GT(\vec y)} \prod_{k=1}^{\ell-1} \prod_{i=1}^{\ell-k} e^{
B_{k+1}(z_i^{\ell-k})-B_k(z_i^{\ell-k})}\mathrm d z_i^{\ell-k}.
\label{eq:densityarrivalell}
\end{equation} 
This leads to 
\begin{equation}
    Z^{\rm stat}_{\ell}(\vec y) =\int_{GT(\vec y)} \prod_{k=1}^{\ell-1} \prod_{i=1}^{\ell-k} e^{
B_{k+1}(z_i^{\ell-k})-B_k(z_i^{\ell-k})}\mathrm d z_i^{\ell-k}.
	\label{eq:defZstatGuillaume}
\end{equation}
Reordering the products, we obtain 
\begin{equation}
Z^{\rm stat}_\ell(\vec y)= \int_{GT(\vec y)} \prod_{k=1}^{\ell} \prod_{i=1}^{k} e^{
	 B_{\ell-k+1}(z_i^{k})- B_{\ell-k+1}(z_{i-1}^{k-1})}\prod_{k=1}^{\ell-1} \prod_{i=1}^{k} \mathrm d z_i^{k},
	 \label{eq:defZstatproof}
\end{equation}
as in \eqref{eq:defZstat}.
\begin{remark}
Averaging \eqref{eq:densityarrivalell} over the Brownian motions, we find, after simplifications valid for $0\le y_1\le \dots\le y_\ell$, that the right hand side equals $e^{\frac{1}{2}\sum_{i=1}^{\ell}y_i}$. However, it would be more interesting to perform this averaging after normalization. 
\end{remark}

\subsection{The O'Connell-Yor semi-discrete polymer} 
\label{sec:OYdetails}
It will be useful to recall the definition of the O'Connell-Yor \cite{o2001brownian} semi-discrete directed polymer. Consider independent standard Brownian motions $B_1, B_2,\dots$. Starting and ending point in the semi-discrete polymer will be points of the form  $(x,n)\in \mathbb R\times \mathbb Z$.  
For real numbers $x<y$ and integers $ i\geq j$, we define the partition function  
\begin{equation}
    Z^{B}_1[ (x,i)\vert (y,j)]= \int_{x=z^{0}<z^{1}<\dots<z^{i-j}<z^{i-j+1}=y} \prod_{r=j}^i e^{B_{r}(z^{i-r+1})-B_r(z^{i-r})} \prod_{r=j}^{i-1} dz^{i-r}.
    \label{eq:defZOY}
\end{equation}
The partition function $Z^{B}_1[ (x,i)\vert (y,j)]$ was denoted $Z^{\rm OY}[ (x,i)\vert (y,j)]$ in Section \ref{sec:mainresults}. 
It can be seen as an integral over semi-discrete  up-right paths from $ (x,i) $ to $(y,j)$ (see Fig. \ref{fig:OConnellYor} (left)) of the exponential of the derivative of the Brownian motions, integrated along the path. We recall that the horizontal lines are indexed from top to bottom. More explicitly, it can be rewritten as an integral over non-increasing semi-discrete paths $\pi:[x,y]\to \lbrace j,\dots, i \rbrace$ such that $\pi(x)=i, \pi(y)=j$,  
\begin{equation}
    Z^{B}_1[ (x,i)\vert (y,j)]= \int_{\pi(x)=i}^{\pi(y)=j}  e^{\int_x^y dB_{\pi(t)}(t)} \mathcal D\pi.
    \label{eq:defZOYpath}
\end{equation}

We further  define partition functions for $\ell$ non-intersecting paths. For $i_1\leq \dots \leq i_{\ell}$, $j_1\leq \dots \leq j_\ell$, $x_1\leq \dots\leq x_\ell$, $y_1\leq \dots\leq y_\ell$, we define the partition function 
\begin{equation}
    Z^{B}_{\ell}[(x_1,i_1),\dots, (x_{\ell}, i_{\ell}) \vert  (y_1,j_1),\dots, (y_{\ell}, j_{\ell}) ]= \int  e^{\sum_{k=1}^{\ell} \int_{x_k}^{y_k}  dB_{\pi_k(t)}(t) }   \prod_{k=1}^{\ell}  \mathcal D\pi_k,
    \label{eq:defZOYmulti}
\end{equation}
where the integration is performed over non-intersecting non-increasing semi-discrete paths $\pi_1, \dots, \pi_{\ell}$ such that $\pi_k:[x_k,y_k]\to \lbrace j_k, \dots, i_k\rbrace$ with $\pi_k(x_k)=i_k$ and $\pi_k(y_k)=j_k$. 
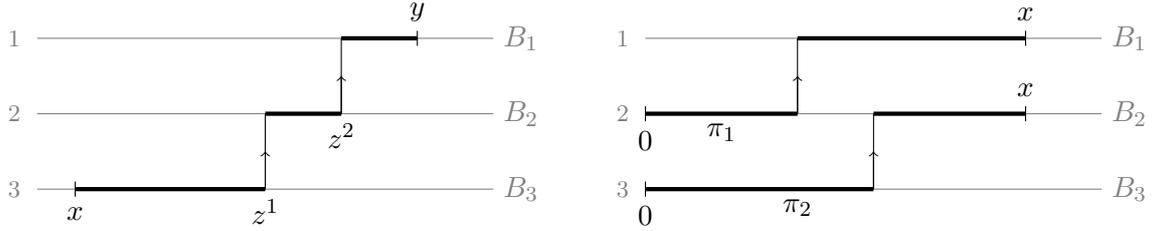
\begin{figure}
	\centering
	\begin{tikzpicture}
		\begin{scope}
\draw[gray] (-0.3,2) node{\footnotesize $1$};
\draw[gray] (-0.3,1) node{\footnotesize $2$};
\draw[gray] (-0.3,0) node{\footnotesize $3$};
\draw[gray] (0,0) -- (6,0) node[anchor=west] {$B_3$};
\draw[gray] (0,1) -- (6,1) node[anchor=west] {$B_2$};
\draw[gray] (0,2) -- (6,2) node[anchor=west] {$B_1$};
\draw (3,0) node[anchor=north] {$z^1$};
\draw (4,1) node[anchor=north] {$z^2$};
\draw (5,1.9)--(5,2.1) node[anchor=south] {$y$};

\draw[->] (3,0) -- (3,0.5);
\draw (3,0) -- (3,1);
\draw[->] (4,1) -- (4,1.5);
\draw (4,1) -- (4,2);
\draw[ultra thick] (0.5,0) -- (3,0);
\draw[ultra thick] (3,1) -- (4,1);
\draw[ultra thick] (4,2) -- (5,2);
\draw (0.5,-0.1) node[anchor=north] {$x$} -- (0.5,0.1);
\end{scope} 

	\begin{scope}[xshift=8cm]
\draw[gray] (-0.3,2) node{\footnotesize $1$};
\draw[gray] (-0.3,1) node{\footnotesize $2$};
\draw[gray] (-0.3,0) node{\footnotesize $3$};
\draw[gray] (0,0) -- (6,0) node[anchor=west] {$B_3$};
\draw[gray] (0,1) -- (6,1) node[anchor=west] {$B_2$};
\draw[gray] (0,2) -- (6,2) node[anchor=west] {$B_1$};
%\fill (0,0) circle(0.05);
\draw (0,-0.1) node[anchor=north] {$0$} -- (0,0.1);
\draw (0,0.9) node[anchor=north] {$0$} -- (0,1.1);
%\fill (3,0) circle(0.05);
\draw (2,0) node[anchor=north] {$\pi_2$};
%\fill (2,1) circle(0.05);
\draw (1,1) node[anchor=north] {$\pi_1$};
%\fill (5,1) circle(0.05);
\draw (5,0.9)--(5,1.1) node[anchor=south] {$x$};
%\fill (5,2) circle(0.05);
\draw (5,1.9) -- (5,2.1) node[anchor=south] {$x$};
\draw[->] (3,0) -- (3,0.5);
\draw (3,0) -- (3,1);
%\draw[->] (0,0) -- (0,0.5);
%\draw (0,0) -- (0,1);
\draw[->] (2,1) -- (2,1.5);
\draw (2,1) -- (2,2);
%\draw[->] (5,1) -- (5,1.5);
%\draw (5,1) -- (5,2);
\draw[ultra thick] (0,0) -- (3,0);
\draw[ultra thick] (0,1) -- (2,1);
\draw[ultra thick] (3,1) -- (5,1);
\draw[ultra thick] (2,2) -- (5,2);
\end{scope} 
\end{tikzpicture}
\caption{A graphical interpretation of the O'Connell-Yor semi-discrete polymer partition function. Left: One polymer from $(x,i)$ to $(y,j)$ with $i=3$ and $j=1$. Right: Two non-intersecting semi-discrete polymers from $(0,2)$ to $(x,1)$ and from $(0,3)$ to $(x,1)$, whose partition function equals $\exp(\mathcal WB_1(x)+ \mathcal WB_2(x))$ defined in \eqref{eq:defWBi}.} 
\label{fig:OConnellYor}
\end{figure}
Hence, the partition function $Z^{\rm stat}_{\ell}(\vec y)$ defined in \eqref{eq:defZstatGuillaume} can be rewritten as a partition function over $\ell$ non-intersecting semi-discrete polymers as 
\begin{equation}
     Z^{\rm stat}_{\ell}(\vec y) = Z^{B}_{\ell}[(0,1),\dots,(0,\ell) \vert (y_1,1), \dots, (y_{\ell},1)].
     \label{eq:relationZstatavecZOY}
\end{equation}
Further, by the Karlin-McGregor formula, this can be rewritten as 
\begin{equation}
     Z^{\rm stat}_{\ell}(\vec y) = \det\left( Z^{B}_1[(0,i) \vert (y_j,1)] \right)_{i,j=1}^{\ell}.
\end{equation}

\subsection{Partition function at arbitrary arrival points (second derivation of the main result)} 
\label{sec:alternativemethod}
In this Section, we show that we can rederive the result \eqref{eq:multipolymers} directly from Proposition \ref{prop:log-gammaarbitrarymultipoint}. Consider the partition function $\Zcal^{stat}_{\ell}(\mathbf p,m)$ defined in \eqref{eq:defZstatp}. As in Section \ref{sec:discretetoKPZ}, we scale $m=tn/2+\ell$ and $p_i= tn/2+x_i\sqrt{n}$.  Let us first consider the initial condition.  When $t=0$, we are considering the partition function $\Zcal^{stat}_{\ell}(\mathbf p,\ell)$, which  converges at large scale to $Z_{\ell}^{\rm stat}(\vec x)$, as defined in \eqref{eq:defZstat} (compare Figure \ref{fig:Zstatmultipoint2} with Figure \ref{fig:OY}). Hence, using the convergence \eqref{eq:convO'ConnellWarren},   for arbitrary $t$,  the partition function $\Zcal^{stat}_{\ell}(\mathbf p,m)$, under appropriate renormalization,  converges (jointly in $\vec x$) to the process 
\begin{equation} 
\int_{\mathbb W_{\ell}} Z_{\ell}^{\rm stat}(\vec y) Z_{\ell}(\vec y,0\vert \vec x,t) d\vec y. 
\label{eq:defZstatinitial}
\end{equation} 
Then, taking the large $n$ limit in the statement of Proposition \ref{prop:log-gammaarbitrarymultipoint}, we obtain that spatial ratios of  \eqref{eq:defZstatinitial} are distributed, for any $t$, as spatial ratios of the process $Z_{\ell}^{\rm stat}(\vec y)$. Letting $t$ to infinity, this yields another derivation of \eqref{eq:multipolymers}. 
 In other terms, the process $Z_{\ell}^{\rm stat}(\vec y)$ is a stationary measure for the Markov process defined by the stochastic PDE \eqref{eq:SPDE}. Note that we expect that there exists other stationary measures, in particular the processes $Z_{\ell}^{\rm stat}(\vec y; \vec a)$ considered in \eqref{eq:mainresultwithdrifts}. 

\subsection{A remarkable property of the gRSK correspondence}
\label{sec:identity}
Given $\ell$ Brownian motions $B_1, \dots, B_\ell$, let us define processes $\mathcal W B_1(x), \dots, \mathcal WB_\ell(x)$ (we follow notations close to \cite{dauvergne2018directed, corwin2020invariance}) such that 
\begin{equation}
    \mathcal WB_1(x) = \log Z^{B}_1[(0,\ell)\vert (x,1)]
\end{equation}
and for all $1\leq k\leq \ell$, 
\begin{equation}
    \mathcal W B_1(x)+ \dots+  \mathcal W B_{k}(x) = \log Z^{B}_{k}[(0,\ell-k+1),\dots, (0,\ell)\vert  (x,1), \dots, (x,k)].
    \label{eq:defWBi}
\end{equation}
For example, $e^{\mathcal W B_1(x)+\mathcal W B_2(x)}$ corresponds to the partition function of two non-intersecting paths in Fig. \ref{fig:OConnellYor} (right).  We may now define partition functions  $Z_{\ell}^{\mathcal WB}$ analogous to $Z^{B}_{\ell}$ after replacing the Brownian motions $B_i$ by the processes $\mathcal WB_i$ (see Fig. \ref{fig:newpartitionfunction}), that is we define 
\begin{equation}
    Z^{\mathcal W B}_{\ell}[(x_1,i_1),\dots, (x_{\ell}, i_{\ell}) \vert  (y_1,j_1),\dots, (y_{\ell}, j_{\ell}) ]= \int e^{\sum_{k=1}^{\ell} \int_{x_k}^{y_k} d\mathcal WB_{\pi_k(t)}(t) }  \prod_{k=1}^{\ell} \mathcal D\pi_k, 
\end{equation}
where we integrate over non-intersecting paths exactly  as in \eqref{eq:defZOYmulti}. We also define the short-hand notation 
\begin{equation}
    Z^{\mathcal WB}_{\ell}(\vec y) = Z^{\mathcal W B}_{\ell}[ (0,1), \dots, (0,\ell) \vert (y_1,1), \dots (y_\ell,1)].
\end{equation}
Then, we have the surprising identity 
\begin{equation}
    Z^{\mathcal WB}_{\ell}(\vec y) = Z^{\rm stat}_\ell(\vec y). 
    \label{eq:surprisingidentity}
\end{equation}
This is a consequence of \cite[Theorem 1.1]{corwin2020invariance} (after taking a straightforward limit to transform discrete polymers in semi-discrete ones). This is also a positive temperature analogue of \cite[Prop 4.1]{dauvergne2018directed}.  While $Z^{\rm stat}$ is a simpler model, the partition function $Z^{\mathcal W B}$ may be easier to understand in some cases as the geometry of optimizing paths is simpler. In particular, this observation was crucial in the rigorous construction of the Airy sheet \cite{dauvergne2018directed}. This type of identity goes back to the work of  \cite{biane2005littelmann} (see the discussion in \cite{corwin2020invariance}). 

\begin{figure}
	\centering
	\begin{tikzpicture}
	\begin{scope}
\draw[gray] (-0.3,2) node{\footnotesize $1$};
\draw[gray] (-0.3,1) node{\footnotesize $2$};
\draw[gray] (-0.3,0) node{\footnotesize $3$};
\draw[gray] (0,0) -- (7,0) node[anchor=west] {$\mathcal W B_3$};
\draw[gray] (0,1) -- (7,1) node[anchor=west] {$\mathcal W B_2$};
\draw[gray] (0,2) -- (7,2) node[anchor=west] {$\mathcal W B_1$};
\draw (0.5,2) node[anchor=north] {$\pi_1$};
\draw (0.5,1) node[anchor=north] {$\pi_2$};
\draw (0.5,0) node[anchor=north] {$\pi_3$};
\draw (5.3,0) node[anchor=north] {$z_1^1$};
\draw (3.6,1) node[anchor=north] {$z_1^2$};
\draw (5.7,1) node[anchor=north] {$z_2^2$};
\draw (4,1.9)-- (4,2.1) node[anchor=south] {$y_2$};
\draw (6,1.9) -- (6,2.1) node[anchor=south] {$y_3$};
\draw (1,1.9) -- (1,2.1) node[anchor=south] {$y_1$};
\draw[->] (5.3,0) -- (5.3,0.5);
\draw (5.3,0) -- (5.3,1);
\draw[->] (3.6,1) -- (3.6,1.5);
\draw (3.6,1) -- (3.6,2);
\draw[->] (5.7,1) -- (5.7,1.5);
\draw (5.7,1) -- (5.7,2);
\draw[ultra thick] (0,0) -- (5.3,0);
\draw[ultra thick] (0,1) -- (3.6,1);
\draw[ultra thick] (5.3,1) -- (5.7,1);
\draw[ultra thick] (0,2) -- (1,2);
\draw[ultra thick] (3.6,2) -- (4,2);
\draw[ultra thick] (5.7,2) -- (6,2);
\draw (0,-0.1) node[anchor=north] {$0$} -- (0,0.1);
\end{scope} 
	\end{tikzpicture}
\caption{The definition of the partition function $Z^{\mathcal W B}(\vec y)$ in terms of non-intersecting paths with fixed endpoints.
The definition is analogous to the definition of $Z^{\rm stat}(\vec y)$ in Figure \ref{fig:OY} after replacing the Brownian motions $B_1, \dots, B_{\ell}$ by the processes $\mathcal WB_1, \dots, \mathcal WB_{\ell}$ defined in \eqref{eq:defWBi}. Surprisingly, $Z^{\rm stat}(\vec y)$ and $Z^{\mathcal W B}(\vec y)$ are equal (see \eqref{eq:surprisingidentity})}
\label{fig:newpartitionfunction}
\end{figure}
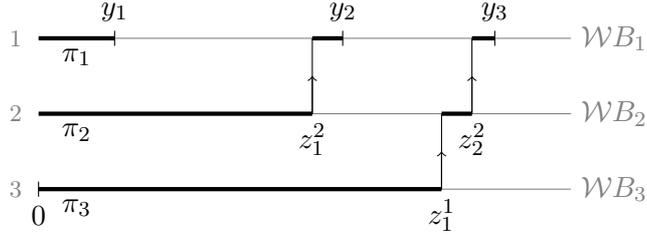

\subsection{Large-scale limit} 

In the limit where all points $y_1$ go to infinity and are well separated, we obtain an alternative expression for the partition function $Z_{\ell}^{\rm stat}$ in terms of Dyson Brownian motion.  First, consider the large scale limit of the process $\mathcal W B_1, \dots, \mathcal W B_{\ell}$. As $y\to+\infty$, one knows that
\begin{equation} 
\left (\frac{1}{\sqrt{x}} \mathcal W B_k(xy) \right)_{1\leq k\leq \ell, y>0} \xRightarrow[x\to\infty]{(d)} \left( \Lambda_k(y) \right)_{1\leq k\leq \ell, y>0}
\label{eq:zerotemperature}
\end{equation}  
where $\Lambda_1(y) > \Lambda_2(y) > \dots > \Lambda_\ell(y) $ are the ordered eigenvalues of a GUE($\ell$)
random matrix $H$ with measure $\sim \exp(-\frac{1}{2 y} {\rm Tr} H^2)$, and as a process in the variable  $y$ perform the Dyson Brownian motion. This convergence \eqref{eq:zerotemperature} holds jointly for $1\leq k\leq \ell$ and as a process in the variable  $y$ \cite{o2002representation} (see also \cite{o2002random}). It was first discovered in \cite{baryshnikov2001gues, gravner2001limit} in the special case $y=1,k=1$.  
The process $\Lambda_1, \dots, \Lambda_{\ell}$ can also be seen as $\ell$ Brownian motions started at $\Lambda_j(0)=0$ and conditioned not to intersect for all $y>0$.

Now, we consider the limit $\frac{1}{\sqrt{x}} \log Z^{\mathcal W B}_{\ell}(x\vec y)$ as $x$ goes to infinity. We have that 
\begin{align} 
    \lim_{x\to\infty} \frac{1}{\sqrt{x}} \log Z^{\rm stat}_{\ell}(x\vec y) &= \lim_{x\to\infty} \frac{1}{\sqrt{x}} \log Z^{\mathcal W B}_{\ell}(x\vec y) \\ 
    &= \sup_{z \in GT(\vec y)} \left\lbrace \sum_{k=1}^{\ell}\sum_{i=1}^{k} \Lambda_{\ell-k+1}(z_i^k)-\Lambda_{\ell-k+1}(z_{i-1}^{k-1})  \right\rbrace.
    \label{eq:supremumtosimplify}
\end{align}
where the supremum is over Gelfand-Tsetlin patterns as defined in \eqref{eq:defGT}. 

\appendix

\newpage
\addcontentsline{toc}{section}{Appendix: Maximum and argmax for the $\mathrm{GUE}(2)$ Dyson Brownian motion with drift}
\section*{Appendix: Maximum and Argmax for the $\mathrm{GUE}(2)$ Dyson Brownian motion with drift}

Let us recall that for a Brownian motion with $-1$ drift, i.e. for the process
$B(t)-t$, where $B(t)$ is a standard Brownian with $B(0)=0$, the joint PDF $p(t_m,m)$ of the maximum $m=\max_{t \geq 0} \left\lbrace B(t)-t\right\rbrace$ 
and of the time $t_m= {\rm argmax}_{t \geq 0} \left\lbrace B(t)-t\right\rbrace$ when the maximum is achieved,  is a classical
result : see e.g.  \cite[Chapter IV, item 32]{borodin1996handbook} 
or taking the limit $T \to +\infty$ in \cite{shepp1979joint} in \cite[Eq. (30)]{majumdar2008optimal}. Its Laplace transform reads
\be 
\int_0^{+\infty} dt e^{- s t - a m} p(t,m) = \frac{2}{\sqrt{2 s +1} + a + 1} ,
\ee 
which leads to the PDF of the time when the maximum is achieved as
\be 
p(t) = \int_0^{+\infty} dm \, p(t,m)= \sqrt{\frac{2}{\pi t}} e^{-t/2} - {\rm Erfc}\left(\sqrt{\frac{t}{2}}\right).
\ee 
We used these results recently in \cite{barraquand2021kardar} to obtain the scaled endpoint distribution
for a very long polymer in the bound phase near a wall, and near the unbinding transition
(i.e. using the Brownian half-space stationary measure). 

Here we will obtain the joint PDF $p(t_m,m)$ for the two analog quantities for the $\mathrm{GUE}(N)$ Dyson Brownian motion, which we denote $\mathrm{DBM}(N)$,
in the simplest case of $N=2$, i.e. with two particles. With a slight change of notations from the text, we define 
\be 
m = {\rm max}_{t \in \mathbb{R}_+} ( \Lambda_1(t) - t),  \quad \quad 
t_m = {\rm argmax}_{t \in \mathbb{R}_+} ( \Lambda_1(t) - t) 
\ee 
where $\Lambda_1(t)>\Lambda_2(t)$ is the $\mathrm{DBM}(2)$ defined in the text for $\ell=2$ (normalized to be locally a standard
Brownian motion). As discussed in Section \ref{sec:applications}, the marginal distribution for $t_m$ gives the scaled PDF of the distance between
the endpoints of two very long non-intersecting polymers in full space, when the starting points of polymers paths are sent to infinity with a common angle $b$, in the
limit of a small angle $b \to 0$. We use here an extension of the method used in \cite{majumdar2008optimal} for
the single Brownian problem. Note that some properties of the extrema of the DBM, or of non-crossing Brownians,
have been studied before, see e.g. \cite{borodin2009maximum,rambeau2011distribution,nguyen2017extreme,gautie2019non}, but to our knowledge the observable needed here was not studied.

Let us start by denoting $G_m(x,z,t)$ the probability density that a single Brownian with drift $-1$,
i.e. the process $B(t)-t$, with $B(0)=z$, ends up at $x$ at time $t$ and remains below level $m>z$
during that time, i.e. 
\be 
G_m(x,z,t) dx = \mathbb{P}\left( B(t) - t \in [x,x+dx] \text{ and }  \max_{\tau \in [0,t]} \left\lbrace B(\tau) - \tau \right\rbrace < m \bigg\vert B(0)=z\right).
\ee 
From the images method (see \cite{majumdar2008optimal} for a simple derivation) the explicit 
formula is given by 
\be 
G_m(x,z,t) = \frac{1}{\sqrt{2 \pi  t} } e^{- \frac{t}{2} - (x-z) }
\left( e^{- \frac{(x-z)^2}{2  t} } - e^{- \frac{(2 m -x-z)^2}{2  t} } \right).
\ee 

Let us now consider the process $\Lambda_1(t) -t > \Lambda_2(t) -t > \dots$ i.e. the $\mathrm{DBM}(N)$ 
to which one adds a uniform drift $-1$. Let us denote $P_m^{\rm DBM}(\vec x, \vec z, t)$ the propagator of
this process, i.e. the probability density that if it starts at $\vec z=(z_1>\dots>z_N)$, it will end up 
at $\vec x=(x_1>\dots>x_N)$ and will remain below level $m$ for all times up to $t$. This propagator is given by the formula
\be 
P_m^{\rm DBM}(\vec x, \vec z, t) = \frac{\Delta(\vec x)}{\Delta(\vec z)} \det\left(G_m(x_i,z_j,t)\right)_{i,j=1}^N 
\ee 
For $m=+\infty$ this is a well known formula, and one can see that it extends in presence of a boundary and a uniform drift. 
Integrating over $\vec x$ with $x_1>\dots>x_N$ gives the total probability that the drifted $\mathrm{DBM}(N)$ remains below level $m$.
From now on we restrict to $N=2$, but similar formulae can be derived for any $N$. 

Following similar arguments as in \cite{majumdar2008optimal} one obtains that the joint probability 
density $p(t_m,m)$ for the maximum and the time to the maximum is equal to the small $\epsilon$
limit 
\begin{align} 
 p(t_m,m) &= \lim_{\epsilon \to 0} \frac{c}{\epsilon^2} \lim_{T \to +\infty} p\left(t=t_m,m,x_1=m-\epsilon,T\right) ,\\
p(t,m,x_1,T) &= \int_{m>y_1>y_2} dy_1 dy_2 \int_{x_2<x_1} dx_2 \, 
P_m^{\rm DBM}(\vec x, \vec z, t) \, P_m^{\rm DBM}(\vec y, \vec x, T-t), \label{prod} 
\end{align} 
where $c$ is a normalizing constant.
The large $T$ limit exists and is finite because of the negative drift.

Let us look first at the second factor in \eqref{prod}. We need to perform the integral in variables  $y_{12}:=y_1-y_2$ and $y_{m1}:=m-y_1$, 
each integrated  over $\mathbb{R}_+$. It is convenient to observe that at large $T$,  the double integral is dominated by 
the typical region (in presence of the drift), i.e
$y_{12}=T^{1/2} \tilde y_{12}$, $y_{m1}=T + T^{1/2} \tilde y_{m1}$, where 
$\tilde y_{12}=O(1)$ is integrated over $\mathbb{R}_+$ and $\tilde y_{m1}=O(1)$
integrated over $\mathbb{R}$. Performing these integrals we obtain
\be 
\lim_{T \to +\infty} \int_{m>y_1>y_2} dy_1 dy_2  \, P_m^{\rm DBM}(\vec y, \vec x , T-t) |_{x_1=m - \epsilon}
= \frac{4 \epsilon e^{-x_{12}}}{x_{12}}  (x_{12} \cosh x_{12} - \sinh x_{12}) + o(\epsilon), 
\ee 
where $x_{12}=x_1-x_2>0$. Note that it is independent of both $m$ and $t$. 

Now, the first factor in \eqref{prod} needs to be computed setting
$z_1=\eta$, $z_2=0$ in the limit $\eta \to 0$. One obtains
\be 
\lim_{\eta \to 0^+} P_m^{\rm DBM}(\vec x, (\eta,0) , t) \Big\vert_{x_1=m - \epsilon} \!\!= \frac{\epsilon x_{12}}{\pi  t^3} e^{-\frac{2 x_{12} (m-t)+2
   (m+t)^2+x_{12}^2}{2 t}} \left(e^{\frac{2 m
   x_{12}}{t}} \left(m x_{12}-t\right)+m
   x_{12}+t\right).
\ee 
Putting all together and performing the remaining integral over $x_{12}>0$, considerable simplifications occur 
and we finally obtain (i) the joint PDF $p(t,m)$ for the
maximum and time of maximum and (ii) the PDF $p(t)$ for the time of
maximum (the normalizing constant is found to be equal to $c=1/2$) 
\begin{align}
 p(t,m) &=  \frac{2 \sqrt{\frac{2}{\pi }} e^{-\frac{m^2+4 m
   t+t^2}{2 t}} \left(\left(m^2 (2 t-1)-t^2+t\right)
   \sinh (m)+m \left(m^2+(t-1) t\right) \cosh
   (m)\right)}{t^{5/2}}, \nonumber \\
 p(t) &= \int_0^{+\infty} dm p(t,m) =
2 \left(\text{erf}\left(\frac{\sqrt{t}}{\sqrt{2}}\right)-e^{4 t} (24 t-1) \text{erfc}\left(\frac{3
   \sqrt{t}}{\sqrt{2}}\right)+8 \sqrt{\frac{2}{\pi }}
   e^{-t/2} \sqrt{t}-1\right) \label{pt} .
\end{align}
with $\int_0^{+\infty} dt p(t)=1$. One also finds that the Laplace transform, for $s\geq 0$, has the expression
\be 
\int_0^{+\infty} dt e^{- s t} p(t) = \frac{4 \left(s \left(3 s-10 \sqrt{2 s+1}+22\right)-8
   \sqrt{2 s+1}+8\right)}{(s-4)^2 s \sqrt{2 s+1}}.
\ee 
This leads to the first few moments $\int_0^{+\infty} dt t^p p(t)= \{\frac{5}{4}, \frac{29}{8}, \frac{555}{32}, \frac{3747}{32} \}$ for
$p=1,2,3,4$. One also notes that the double Laplace transform of $p(t,m)$ has a relatively simple form
\be 
\int_0^{+\infty} dt e^{- s t - a m} p(t,m) =
\frac{16 \left(a+\frac{3 s}{\sqrt{2
   s+1}}+\frac{2}{\sqrt{2 s+1}}+2\right)}{\left(2 a
   \sqrt{2 s+1}+(a+2)^2+2 s+4 \sqrt{2 s+1}\right)^2}.
\ee

\begin{footnotesize} 
\bibliographystyle{unsrt}
\bibliography{NonIntersectingPolymers.bib} 
\end{footnotesize}

\end{document}